\documentclass[USenglish]{lipics}

\usepackage{microtype}
\usepackage[svgnames]{xcolor}
\usepackage{amssymb}
\usepackage{amsmath}
\usepackage{colonequals}
\usepackage{stmaryrd}
\usepackage{xfrac}
\usepackage{wrapfig}
\usepackage{tikz}
\usepackage{bm}

\hypersetup{
    colorlinks,
    linkcolor={red!50!black},
    citecolor={blue!50!black},
    urlcolor={blue!30!black}
}

\usepackage[ruled,vlined,linesnumbered]{algorithm2e}

\usepackage[compress,sort,noadjust]{cite}

\theoremstyle{plain}

\newtheorem{observation}[theorem]{Observation}

\usepackage{apxproof}
\def\appendixbibliographystyle{alphaabbr}
\newtheoremrep{theorem}{Theorem}
\newtheoremrep{proposition}[theorem]{Proposition}
\newtheoremrep{lemma}[theorem]{Lemma}
\newtheoremrep{corollary}[theorem]{Corollary}

\title{Top-k Querying of Unknown Values\protect\\ under Order Constraints (Extended Version)}
\titlerunning{Top-k Querying of Unknown Values under Order Constraints (Extended Version)}

\author[1]{Antoine Amarilli}
\author[2]{Yael Amsterdamer}
\author[3]{Tova Milo}
\author[4,5]{Pierre~Senellart}
\affil[1]{LTCI, Télécom ParisTech, Université Paris-Saclay;
Paris, France\\\texttt{first.last@telecom-paristech.fr}}
\affil[2]{Bar Ilan University; Ramat Gan, Israel\\\texttt{first.last@biu.ac.il}}
\affil[3]{Tel Aviv University; Tel Aviv, Israel\\\texttt{last@cs.tau.ac.il}}
\affil[4]{DI, École normale supérieure, PSL Research University; Paris,
France\\\texttt{first.last@ens.fr}}
\affil[5]{Inria Paris; Paris, France}

\Copyright{Antoine Amarilli; Yael Amsterdamer; Tova Milo; Pierre Senellart}

\colorlet{darkred}{red!50!black}

\newcommand{\oof}[1]{\operatorname{O}\!\left({#1}\right)}
\newcommand{\oofb}[1]{\operatorname{O}({#1})}
\renewcommand{\leq}{\leqslant}
\renewcommand{\geq}{\geqslant}
\renewcommand{\epsilon}{\varepsilon}

\newcommand{\deft}[1]{\emph{#1}}

\DeclareMathOperator{\pw}{pw}

\newcommand{\calC}{\mathcal{C}}

\newcommand{\calP}{\mathcal{P}}

\newcommand{\calS}{\mathcal{S}}
\newcommand{\calT}{\mathcal{T}}

\newcommand{\calX}{\mathcal{X}}

\newcommand{\quot}[2]{#1/{#2}}
\newcommand{\card}[1]{\left|{#1}\right|}
\newcommand{\cardb}[1]{|{#1}|}

\newcommand\restr[2]{{
  \kern-\nulldelimiterspace 
  #1 
  _{|#2} 
  }}

\newcommand{\pvol}{V} 
\renewcommand{\d}{\,\mathrm{d}} 

\newcommand{\defeq}{\colonequals}

\newcommand{\fp}{\mathrm{FP}}
\newcommand{\sharpp}{\mathrm{\#P}}
\newcommand{\fpsp}{\fp^{\sharpp}}

\renewcommand{\a}{\mathrm{a}}
\renewcommand{\b}{\mathrm{b}}
\renewcommand{\c}{\mathrm{c}}
\renewcommand{\d}{\mathrm{d}}
\newcommand{\e}{\mathrm{e}}
\newcommand{\f}{\mathrm{f}}
\newcommand{\h}{\mathrm{h}}
\renewcommand{\l}{\mathrm{l}}
\renewcommand{\r}{\mathrm{r}}
\newcommand{\s}{\mathrm{s}}

\newcommand{\uniform}{{\sf Uniform}\xspace}
\newcommand{\stable}{{\sf Stable}\xspace}

\newcommand{\Xexact}{\calX_{\mathrm{exact}}}

\begin{document}

\maketitle

\setcounter{footnote}{0}

\begin{abstract}
Many practical scenarios make it necessary to evaluate top-$k$ queries over data
items with partially unknown values. This paper considers a setting where the
values are taken from a numerical domain, and where some \emph{partial order
constraints} are given over known and unknown values: under these constraints,
we assume that all possible worlds are equally likely. Our work is the first to
propose a principled scheme to derive the value distributions and expected
values of unknown items in this setting, with the goal of computing estimated
top-$k$ results by interpolating the unknown values from the known ones. We
study the complexity of this general task, and show tight complexity bounds,
proving that the problem is intractable, but can be tractably approximated. We
then consider the case of tree-shaped partial orders, where we show a
constructive PTIME solution. We also compare our problem setting to other
top-$k$ definitions on uncertain data.

\end{abstract}

\section{Introduction} \label{sec:intro}

\noindent Many data analysis tasks involve queries over ordered data, such as maximum and
top-$k$ queries, which must often be evaluated in presence of
\emph{unknown data values}.
This problem occurs in many real-life scenarios: retrieving or computing
exact data values is often expensive, but querying the partially unknown data may
still be useful to obtain approximate results, or to decide which data values
should be retrieved next. In such contexts, we can often make use of \emph{order constraints}
relating the data values, even when they are unknown: for instance, we know that object $A$
is
preferred to object $B$ (though we do not know their exact rating).

This paper thus studies the following general problem. We consider a set
of numerical values, some of which are unknown, and we assume a
\emph{partial order} on these values:  we may know that $x \geq y$ should hold although the values $x$ or $y$ are unknown. Our goal is to \emph{estimate} the unknown values, in a principled way, and to 
evaluate top-$k$ queries, namely find the items with (estimated) highest values.

{Without further information, one may
assume that every valuation compatible with the order constraints is
equally likely, i.e., build a probabilistic model where valuations are
\emph{uniformly distributed}. Indeed, uniform distributions in the
absence of prior knowledge is a common assumption in probabilistic data
management~\cite{abiteboul2011capturing,cheng2003evaluating,lian2008probabilistic}
for continuous
distributions on data values within an interval;
here we generalize to a uniform distribution over multiple unknown values.
Though the distribution is uniform, the 
dependencies between values lead to
non-trivial insights about unknown values and top-$k$ 
results, as we will
illustrate.}

\begin{figure}
\centering
\includegraphics[width=0.65\linewidth]{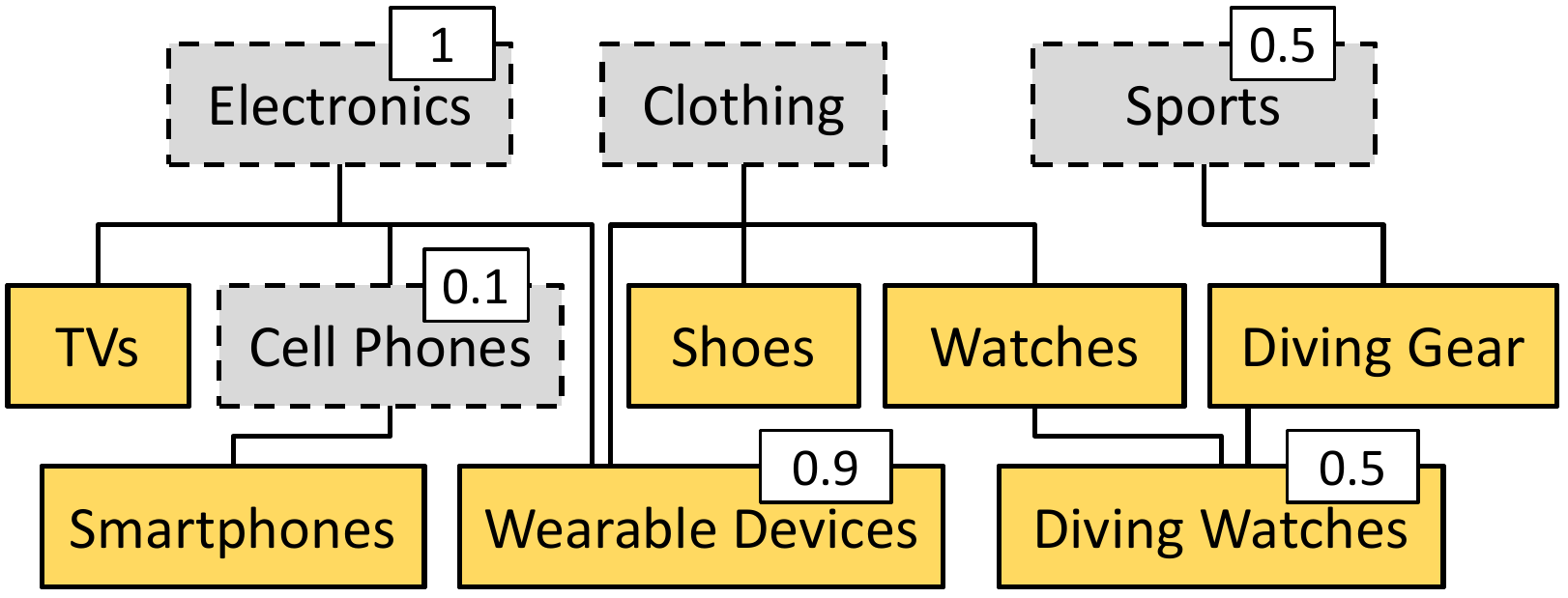}
\caption{Sample catalog taxonomy with compatibility scores}
\label{fig:catalog}
\end{figure}

\subparagraph*{Illustrative example.}
We consider a specific application setting where our problem occurs.
Consider a scenario where products are classified in a catalog taxonomy (Figure~\ref{fig:catalog})
using human input: 
{the relevance of a product to any category is captured by a
\deft{compatibility score}. Assessing this compatibility is often left to human judgement rather than attempting to derive it from records or statistics~\cite{chimera,BraggMW13,parameswaran2011human}. Thus we assume scores are obtained via questions to domain experts or to a 
	crowd
        of unqualified users\footnote{In
		the latter case, we aggregate the answers of multiple workers to
                obtain the compatibility score.}.} 
Our goal is to assign the product to the top-$k$ most compatible categories
among a set of \emph{end categories} (in yellow with a solid border), as
opposed to \emph{virtual categories} (dashed border). The virtual categories generalize the end
categories, and allow us to ask broader questions to experts, but we 
do not try to assign products to them, e.g., they do not have a dedicated page in our online store.

Imagine now that the product to classify is a \emph{smartwatch}, and that we want to find the top-2 end categories for it.
We asked an expert for its compatibility score with some categories (both end and
virtual categories), which we indicate in Figure~\ref{fig:catalog}.
Because expert input is costly, however, we wish to choose
the top-2 end categories based on the incomplete information that we have.
The na\"{i}ve answer is to look only at categories with known scores, and to
identify \textsf{Wearable Devices} and \textsf{Diving Watches} as the best end
categories.

{In this scenario, however, we can impose a natural partial order over
the scores, both known and unknown:
any product that belongs to a specific category (e.g., \textsf{Smartphones})
conceptually also belongs to each of the more general categories (e.g.,
\textsf{Cell Phones}). We can thus require that if a category $x$ is a sub-category of~$y$, then
the product's compatibility score for~$x$ should be at most its
score for~$y$. This can be explained in the instructions to the humans providing
the scores, and enforced by the user interface. Beyond this constraint, we do
not make any further assumption on the scores.}

Now, note that due to order constraints, 
the scores of
\textsf{Watches} and \textsf{Diving Gear}, while unknown, cannot be lower than that of \textsf{Diving Watches}; so either of the two could replace
\textsf{Diving Watches} in the top-2 answer. To choose between these categories, we observe that the score of \textsf{Diving Gear} must be
exactly~$0.5$ (which bounds it from above and below).
In contrast,
as the score of \textsf{Wearable Devices} is~$0.9$,
\textsf{Clothing} has a score of at least~$0.9$, so the score of
\textsf{Watches} can be anything between~$0.5$ and an unknown value which is
$\geq 0.9$. A better top-2 answer is thus
\textsf{Wearable Devices} and \textsf{Watches}, the latter being likely to have
a higher score than \textsf{Diving Watches} (or \textsf{Diving Gear}).

\subparagraph*{Other application domains.}
{Beyond the crowdsourcing example that we illustrate, our setting is
relevant to many other application domains involving different types of unknown
data. For instance, in the domain of Web services~\cite{SolimanIB10}, one may
wish to find the $k$-most relevant apartments for a user, given user-provided
criteria on the apartment (price range, location, size, etc.)\ but no precise way
to aggregate them into a relevance function. In this case, order constraints may
be imposed on the unknown relevance of apartments, whenever one apartment
\emph{dominates} another (i.e., is better for each criterion); exact-value constraints may capture user-provided ratings to viewed apartments; and a top-$k$
query could be used, e.g., to select the most relevant apartments among those available for sale.}

{As another example, in the context of top-$k$ queries over sensor data~\cite{HaghaniMA09,lian2010generic}, one may wish to find the $k$-fastest drivers in a certain region given partial data from speedometers and street cameras; comparing the progress and locations of vehicles may yield partial order constraints on their speed. Other domains include, e.g., data mining~\cite{amarilli2014complexity},
managing preference data~\cite{stoyanovich2013understanding}, or finding optimal
values of black-box functions expressing cost, running time, etc.}

\subparagraph*{Contributions.}
As previously mentioned, we assume a uniform probability distribution over valuations of unknown items, which we capture formally in our model via possible-world semantics. We then use the \emph{expected values} of items as an estimate of their unknown values, for the sake of top-$k$ computation (see Section~\ref{sec:prelim}).
Our work presents three main contributions using this general model, as follows.

First, in Section~\ref{sec:brute} we present a \emph{general and principled scheme to interpolate unknown values from known ones under partial order constraints}, and thereby obtain the top-$k$ such values. We implement this in an algorithm that is polynomial in the number of possible item orderings, and consequently show that in the worst case it is in $\fpsp$ in the size of the input.\footnote{{\#P is the class of counting
	problems that return the number of solutions of NP problems. $\fpsp$ is the
class of function problems that can be computed in PTIME using a \#P
oracle.}}  The problem of finding expected values has a geometric
characterization as centroid computation in high-dimensional polytopes
(as we explain further); however, our $\fpsp$ membership result goes beyond existing computational geometry results since the constraints that we consider, namely, partial order and exact-value constraints, correspond to special classes of polytopes not studied in the context of geometry. Indeed, centroid computation is generally not in $\fpsp$~\cite{lawrence1991polytope,rademacher2007approximating}.
Our work also departs from 
previous work on top-$k$ queries over incomplete or probabilistic
data~\cite{cormode2009semantics,JestesCLY11,SolimanIB10}: we do not make the simplifying assumption that item distributions are independent and given, but rather study the effect of constraints on individual item distributions. 

Our second main contribution, in Section~\ref{sec:complexity}, is to
establish \emph{hardness results}, and specifically, a matching lower
$\fpsp$ bound for top-$k$ computation. While the $\sharpp$-hardness of
computing expected values follows from the geometric characterization
of the problem~\cite{rademacher2007approximating}, we show that top-$k$
is hard even without computing expected values. Hence the $\fpsp$ bound
is tight for both interpolation and top-$k$; this shows that the
assumption regarding variable independence in previous work, that enables
PTIME solutions~\cite{cormode2009semantics,JestesCLY11,SolimanIB10}, indeed simplifies the problem. To complete the picture we discuss possible approximation schemes for the interpolation and top-$k$ problems, again by the connection to centroid computation.

Our third main contribution, in Section~\ref{sec:tractable}, is the study of \emph{tractable cases}, following the hardness of the general problem and the high complexity of approximation.
We devise a PTIME algorithm to compute expected
values and top-$k$ results when the order constraints are \emph{tree-shaped} or decomposable to trees.
This class of constraints is commonly encountered in the context of unknown
values (e.g., taxonomies of products are often trees rather than DAGs); yet, to our knowledge, the corresponding polytopes have no equivalents in the context of computational geometry.

Our results also include 
a review of existing definitions for top-$k$ over uncertain data, which motivates our particular choice of definition
(in Section~\ref{sec:variants}). We survey related work in more depth in
Section~\ref{sec:related} and conclude in Section~\ref{sec:conc}.
Our results are provided with complete proofs, given in appendix for
lack of space.

\section{Preliminaries and Problem Statement}\label{sec:prelim}
This section introduces the formal definitions for the problem that we study in this paper.
We model known and unknown item values as \emph{variables}, and order constraints as \emph{equalities} and \emph{inequalities} over them.
Then we define the possible valuations for the variables via possible-world semantics, and use this semantics to define a uniform distribution where all worlds are equally likely.
The problem of top-$k$ querying over unknown values can then be formally defined with respect to the expected values of variables in the resulting distribution.

\subsection{Unknown Data Values under Constraints}
Our input includes a set $\calX = \{x_1, \ldots, x_n\}$ of variables with unknown values $v(x_1),
\ldots, v(x_n)$, which we assume\footnote{Our results extend to other bounded,
continuous ranges, because we can rescale them to fall in $[0,1]$.}
to be in the range~$[0, 1]$.
We consider two kinds of constraints over them:

\begin{itemize}
\item \textbf{order constraints}, written
  $x_i \leq x_j$ for $x_i, x_j \in \calX$, encoding that $v(x_i) \leq v(x_j)$;
\item \textbf{exact-value constraints} to represent variables with known values,
  written\footnote{The number $\alpha$ is written as a rational number, represented by its numerator and denominator.} $x_i = \alpha$ for $0 \leq
  \alpha \leq 1$ and for $x_i \in \calX$, encoding that
  $v(x_i) = \alpha$.
\end{itemize}

In what follows, a \deft{constraint set} with constraints of both types is typically denoted~$\calC$. 
We assume that constraints in~$\calC$ are \emph{not contradictory} (e.g., we
forbid $x= 0.1$, $y = 0.2$, $y \leq z$, and $z \leq x$), and that they are
\emph{closed under implication}: e.g., if $x=\alpha$, $y = \beta$ are given, and $\alpha \leq \beta$, then $x \leq y$ is implied and thus should also be in $\calC$.
We can check in PTIME that $\calC$ is non-contradictory
by simply verifying that it does not entail a false inequality on exact values (e.g., $0.2\leq 0.1$ as in our previous example).
The closure of $\calC$ can be found in PTIME as a transitive closure computation~\cite{ioannidis1998efficient} that also considers exact-value constraints.
{We denote by $\Xexact$ the subset of $\calX$ formed of variables with exact-value constraints.}

\begin{example}
In the product classification example from the Introduction, a variable $x_i\in\calX$ would represent the compatibility score of the product to the $i$-th category. If the
score is known, we would encode it as a constraint $x_i=\alpha$. In addition,
  $\calC$ would contain the order constraint $x_i\leq x_j$ whenever category~$i$ is a sub-category of $j$ (recall that the score of a sub-category cannot be higher than that of an ancestor category).
\end{example}

\subsection{Possible World Semantics}\label{sec:possible}
The unknown data captured by~$\calX$ and $\calC$ makes infinitely many valuations of $\calX$ possible (including the true one).
We model these options via possible world semantics: a \deft{possible world}~$w$ for a constraint set~$\calC$ over
$\calX=\{x_1,\dots,x_n\}$ is a vector of values
$w = (v_1,
\ldots, v_n) \in [0, 1]^n$, corresponding to setting $v(x_i)
\colonequals v_i$
for all $i$, such that all the constraints of~$\calC$ hold under this
valuation. The set of all possible worlds is denoted by $\pw_\calX(\calC)$, or
 by $\pw(\calC)$ when $\calX$ is clear from context.

Notice that $\calC$ can be encoded as a set of \emph{linear constraints},
i.e., a set of inequalities between linear expressions on $\calX$ and
constants in~$[0,1]$. 
Thus, following common practice in linear programming, the feasible region of a set of linear constraints ($\pw(\calC)$ in our setting) 
can be characterized geometrically as a convex polytope, termed the \emph{admissible polytope}:
writing $n \defeq \card{\calX}$, each linear constraint defines a feasible
half-space of $\mathbb{R}^n$ (e.g., the half-space where $x\leq y$), and the
convex polytope $\pw(\calC)$ is the
intersection of all half-spaces.
In our setting 
the polytope
$\pw(\calC)$ is bounded within $[0,1]^n$, and it is non-empty by our assumption
that $\calC$ is not contradictory. 
With exact-value constraints, or order constraints such as $x_i \leq
x_j$ and $x_j \leq x_i$,
it may be the case that the dimension of this admissible
polytope is less than~$\card{\calX}$. Computing this dimension can easily
be done in PTIME (see, e.g.,~\cite{schrijver1986theory8}). 

\begin{example}
\label{ex:bounded}
Let $\calX=\{x,y,z\}$. If $\calC= \{x\leq y\}$, the admissible
polytope has dimension~$3$ and is bounded by the planes defined by $x=y$, $x=0$,
$y=1$, $z=0$ and $z=1$. If we add to~$\calC$
the constraint $y=0.3$, the admissible polytope is a 2-dimensional
rectangle bounded by $0\leq x\leq 0.3$ and $0\leq z \leq 1$ on the
$y=0.3$ plane.
We cannot add, for example, the constraint $x=0.5$, because $\calC$ would become
contradictory.
\end{example}

\subsection{Probability Distribution}
Having characterized the possible worlds of $\pw(\calC)$, we 
assume a \emph{uniform} probability distribution over $\pw(\calC)$, as indicated
in the Introduction. This captures the case when all possible worlds are equally
likely, and is a natural choice when we have no information about which
valuations are more probable.

Since the space of possible worlds is continuous, we formally define this distribution via a probability density function (pdf), as follows. Let $\calX$ and $\calC$ define a $d$-dimensional polytope
$\pw_\calX(\calC)$ for some integer~$d$. The \deft{$d$-volume} (also called the Lebesgue
measure~\cite{integrationTextbook} on $\mathbb{R}^d$) is a measure for continuous subsets of $d$-dimensional space, which coincides with length, area, and volume for dimensions 1, 2,
and~3, respectively. We denote by $\pvol_d(\calC)$ the $d$-volume of the admissible polytope, or simply $\pvol(\calC)$ when $d$ is the dimension of $\pw(\calC)$.

\begin{definition}\label{def:uniform}
  The \deft{uniform pdf} $p$ 
  maps each possible world $w \in \pw(\calC)$
  to the constant
  $p(w)\colonequals\sfrac{1}{\pvol(\calC)}$.
\end{definition}

\subsection{Top-k Queries}
\label{sec:deftopk}
We are now ready to formally define the main problem studied in this paper, namely, the evaluation of \emph{top-$k$ queries} over unknown data values.
The queries that we consider retrieve the $k$ items that are estimated to have the highest values, \emph{along with their estimated values}, with ties broken arbitrarily.
We further allow queries to apply a \emph{selection operator} $\sigma$ on the items before performing the top-$k$ computation. 
In our example from the Introduction, this is what allows us to select the
top-$k$ categories among only the end categories.
We denote the subset of $\calX$ selected by~$\sigma$ as $\calX_\sigma$.

If all item values are known, the semantics of top-$k$ queries is clear.
In presence of unknown values, however, the semantics must be redefined to determine how the top-$k$ items and their values are estimated.
In this paper, we estimate unknown items by their \emph{expected value over all
possible worlds}, i.e., their expected value according to the uniform pdf $p$
defined above on
$\pw(\calC)$.
This corresponds to \emph{interpolating} the unknown values from the known ones, and then querying the result. 
We use these interpolated values to define the top-$k$ problem as computing the $k$
variables with the highest expected values, but we also study on its own the
interpolation problem of computing the expected values.

To summarize, the two formal problems that we study on constraint sets are:

\begin{description}
\item[Interpolation.] Given a constraint set $\calC$ over $\calX$ and variable $x\in\calX$,
the \deft{interpolation problem} for $x$ is to compute the expected value
of $x$ in the uniform distribution
over
$\pw_\calX(\calC)$. 
\item[Top-$k$.] Given a constraint set $\calC$ over~$\calX$, a
selection predicate~$\sigma$, and an integer~$k$, the \deft{top-$k$ computation problem} is to compute the ordered list of the $k$ maximal expected values of variables in $\calX_\sigma$ (or less if $\card{\calX_\sigma} \leq k$), with ties broken arbitrarily.
\end{description}

We review other definitions of top-$k$ on uncertain data 
in Section~\ref{sec:variants}, where we justify our choice of semantics.

\subparagraph*{Alternate phrasing.}
The Interpolation problem can also be defined geometrically, as the computation of the \emph{centroid} (or \emph{center of mass}) of the admissible
polytope: the point~$G$ such that all vectors relative to $G$ originating at points within
the polytope sum to zero. 
The constraints that we study correspond to a special kind of
polytopes, for which we will design a specific algorithm in the next section,
and derive an $\fpsp$ membership bound which does not hold for general polytopes
(as explained in the Introduction). However, 
the geometric connection will become useful when we study the
complexity of our problem in Section~\ref{sec:hardness}.

\section{An Algorithm for Interpolation and Top-k}\label{sec:brute}

Having defined formally the problems that we study,
we begin our complexity analysis 
by designing an algorithm that computes the expected value of variables.

The algorithm enumerates all possible
orderings of the variables (to be defined formally below), but it is still
nontrivial: we must handle exact-value constraints specifically, and we must
compute the probability of each ordering to determine its weight in the overall expected value computation.
From the algorithm, we will deduce that our interpolation and top-$k$ 
problems are in $\fpsp$.

\subparagraph*{Eliminating ties.}
To simplify our study, we will eliminate from the start the problem of
\emph{ties}, {which will allow us to assume that values in all worlds are totally ordered}.
We say that a possible world $w = (v_1, \ldots, v_n)$ of $\calC$ has a
\emph{tie} if  $v_i = v_j$ for some $i,j$. {Note that
occasional ties, not enforced by $\calC$, have \emph{an overall
probability of~$0$}: intuitively, if the admissible polytope is
$d$-dimensional, then all the worlds where $v_i = v_j$ correspond to a
$(d-1)$-dimensional hyperplane bounding or intersecting the polytope. A
finite set of such hyperplanes (for every pair of variables) has total
$d$-volume~$0$. Since our computations (volume, expected value) involve integrating over possible worlds, a set of worlds with total probability~0 does not affect the result.}

What is left is to consider ties enforced by $\calC$ (and thus having
probability~$1$).
In such situations,
we can rewrite $\calC$ by merging these variables 
to obtain an equivalent constraint set where ties have probability~$0$.
Formally:

\begin{lemmarep}\label{lem:noties}
  For any constraint set\/ $\calC$, we can construct in PTIME a constraint
  set\/~$\calC'$ such that the probability that the possible
  worlds of\/ $\calC'$
  have a tie (under the
  uniform distribution) is zero, and such that any interpolation or top-$k$
  computation problem on\/ $\calC$
  can be reduced in PTIME to the same type of problem on\/ $\calC'$.
\end{lemmarep}

\begin{proof}
	We consider two types of ties:
        \emph{persistent ties}, which are enforced by~$\calC$ and hold in each
        possible world, and 
        \emph{occasional ties}, which are not enforced by $\calC$ and only hold in
        some possible worlds.
	
	We first formally prove that occasional ties have a total probability~0:
        let $d$ be the dimension of $\pw(\calC)$ and $n \defeq \card{\calX}$.
        Assume that $x,y$ have a tie in some worlds and do not have a tie in
        other worlds (in particular, $x=y$ is not implied by $\calC$). If we add
        constraints to enforce a persistent tie, $\pw(\calC\cup\{x\leq y, y\leq
        x\})$ is now at most $(d-1)$-dimensional. Geometrically, this is a projection of the admissible polytope on the $x=y$ hyperplane. The $d$-volume of $\pw(\calC\cup\{x\leq y, y\leq x\})$ is thus~0, i.e., this set of worlds has probability~0. Then, taking the union of $\mathrm{O}(n^2)$ such sets, to obtain all worlds with occasional ties between any pair of variables, the result has total probability~0.  
	
	Next, to handle persistent ties, we define a constraint set where persistently tied variables are replaced by a single variable. Let  $\quot{\calX}{\sim}$ be the set of equivalence classes of $\calX$. We define $\quot{\calC}{\sim}$ to be the constraint set where every
	occurrence of a variable $x_i\in\calX$ is replaced by some representative $s(x_i)$ of its equivalence class. By definition $\quot{\calC}{\sim}$ does not have persistent ties; and there is clearly a bijection between $\pw(\calC)$ and $\pw(\quot{\calC}{\sim})$.
	
	It remains to show that problems on $\calC$ can be reduced (in PTIME) to $\quot{\calX}{\sim}$. The interpolation problem for a variable $x_i$ on $\calC$ clearly reduces to
	interpolation for $s(x_i)$ on $\quot{\calC}{\sim}$, since due to the
        bijection the expected value of $x_i$ on $\calC$ equals that of $s(x_i)$
        on $\quot{\calC}{\sim}$. If no selection is applied, the reduction for
        top-$k$ works by selecting the top-$k$ variables in
        $\quot{\calC}{\sim}$, replacing each $s(x_i)$ by the variables of its
        equivalence class $[s(x_i)]$, and truncating the obtained ranked list to
        length $k$ (recall that ties can be broken arbitrarily). If a selection
        $\sigma$ is applied, then apply top-$k$ only to representatives of
        selected variables from $\calX_{\sigma}$; and replace each
        representative $s(x_i)$ in the top-$k$ result only by the variables from $[s(x_i)]\cap\calX_{\sigma}$. 
\end{proof}

Hence, we assume from now on that ties have zero probability in $\calC$, so that
we can ignore possible worlds with ties without affecting the correctness of our
analysis. Note that this implies that all of our results also hold for \emph{strict inequality constraints}, of the forms $x<y$ and $x \neq y$.

Under this assumption, we first study in Section~\ref{sec:total} the case where
$\calC$ is a total order. We then handle arbitrary $\calC$ by aggregating over possible variable orderings, in Section~\ref{sec:general}.

\subsection{Total Orders}\label{sec:total}
\begin{toappendix}
  \subsection{Total Orders (Section~\ref{sec:total})}
\end{toappendix}

In this section we assume $\calC$ is a total order\/ $\calC^n_1(\alpha, \beta)$
defined as $x_0 \leq x_1 \leq \cdots \leq x_n \leq x_{n+1}$, where $x_0=\alpha$ and $x_{n+1}=\beta$ are variables with exact-value constraints in $\Xexact$.

We first consider \emph{unfragmented} total orders, where $x_1,\dots,x_n\not\in \Xexact$.
In this case, we can show that the expected value of $x_i$, for $1 \leq i \leq n$, corresponds to a \emph{linear interpolation} of the unknown variables between $\alpha$ and $\beta$, namely:
\(
\textstyle \frac{i}{n+1}\cdot(\beta-\alpha)+\alpha
\). This can be shown formally via a connection to the expected value of
the order statistics of samples from a uniform distribution, which
follows a Beta distribution~\cite{gentle2009computational}.

Now consider the case of \emph{fragmented} total orders, 
where $\calC$ is allowed to contain more exact-value constraints
than the ones on~$\alpha$ and~$\beta$. We observe that we can \emph{split}
the total order into \emph{fragments}: by cutting at each variable that has an
exact-value constraint, we obtain sub-sequences of variables which follow an
unfragmented total order. We can then compute the expected values 
of each fragment independently, and compute the total order volume as the
product of the fragment volumes. The correctness of this computation follows from a more general result
(Lemma~\ref{lem:splittingb}) stated and proven in Section~\ref{sec:tractable}.

Hence, given a constraint set $\calC$ imposing a (possibly fragmented) total order, the expected value of $x_i$ can be
computed as follows. If $x_i\in \Xexact$,
analysis is trivial.
Otherwise, we consider the fragment 
$\calC_{p+1}^{q-1}(v_p,v_{q})$ that contains $x_i$; namely, $p$ is the maximal index such that
$0\leq p <i$ and $x_p\in\Xexact$, and $q$ is the minimal
index such that $i < q \leq n +1$ and $x_{q}\in\Xexact$. The expected value
of $x_i$ can then be computed within
$\calC_{p+1}^{q-1}(v_p,v_{q})$ using linear interpolation.

The
following proposition summarizes our findings:
\begin{proposition}\label{prp:total-order-no-exact}
	Given a constraint set $\calC$ implying a total
	order, 
	the expected value 
	of any variable $x_i\in\calX$ can be computed in PTIME.
\end{proposition}

\begin{toappendix}
  We now give the omitted formal details about the ``Fragmented Total Orders'' paragraph. 
  The soundness of splitting can be proved via the possible world semantics:
assume that $x_i\in \Xexact$. 
Then the possible
valuations of $x_1,\dots,x_{i-1}$ are affected only by the constraints in
$\calC$ on $x_1,\dots,x_i$, and similarly for $x_{i+1},\dots,x_{n}$.
  The formal claim is as follows, which follows from Lemma~\ref{lem:splittingb}:

\begin{lemma}
  \label{lem:splitting}
  {Let $\calC$ be a constraint set implying a fragmented total order
  $x_1\leq\dots\leq x_{n}$, and $x_i=v$ a constraint in $\calC$. Let
  $\calC_{\leq i},\calC_{\geq i}\subseteq\calC$ be the constraints restricted to
  the variable subsets $x_0,\dots,x_i$ and $x_i,\dots,x_{n+1}$ respectively.
  Then there is a bijective correspondence between $\pw_{x_{1}\dots
  x_{i}}(\calC_{\leq i})\times\pw_{x_{i}\dots x_{n}}(\calC_{\geq i})$ and
  $\pw(\calC)$.}
\end{lemma}

  The correspondence is defined from the left-hand side to the right-hand side
  as merging the duplicated component for $x_i$ 
  to obtain worlds of exactly $n$ values.
\end{toappendix}

\subsection{General Constraint Sets}\label{sec:general}
We can now extend the result for total orders to
an expression of the expected value 
for a general constraint set $\calC$.
We apply the previous process to each possible total ordering of the variables, and aggregate the results.
To do this, we define the notion of \deft{linear extensions}, inspired by partial order theory:

\begin{definition}
Given a constraint set $\calC$ over $\calX$, we say that a constraint set $\calT$ is a \deft{linear extension} of $\calC$ if (i) $\calT$ is a total order; (ii) the exact-value constraints of $\calT$ are exactly those of $\calC$; and (iii) $\calC\subseteq \calT$, namely every constraint $x\leq y$ in $\calC$ also holds\footnote{The linear extensions of $\calC$ in this sense are thus exactly the linear extensions of the partial order on~$\calX$ imposed by $\calC$: this partial order is indeed antisymmetric because $\calC$ has no ties.}
  in $\calT$.
\end{definition}

\begin{algorithm}[t]
   \SetNlSty{scriptsize}{}{}
   \KwIn {Constraint set $\calC$ on variables $\calX$ with $n \defeq
   \card{\calX}$ where ties have probability~$0$ and where value range $[0,1]$
   is enforced by constraints; variable $x\in\calX$}
   \KwOut {Expected value of $x$}
   \lIf{$x$ is in $\Xexact$, i.e., has an exact-value constraint in $\calC$ to some~$v$}{\Return $v$}
    $\mathbb{E}_{\calC}[x] \leftarrow 0$; $\pvol(\calC)\leftarrow 0$\;
    $m\leftarrow\card{\Xexact}$\;
   Write $x_0 < \cdots < x_{m-1}$ the variables of~$\Xexact$ and $v_0 <\dots
   <v_{m-1}$ their values\; 
   \ForEach{linear extension $\calT$ of $\calC$}{
   	Write $i_0< \dots< i_{m-1}$ the \emph{indices} in $\calT$ of variables
        $x_0, \ldots, x_{m-1}$
        in $\Xexact$\;
        Write $k$ the \emph{index} in~$\calT$ of the variable~$x$\;
        Write $i_j,i_{j+1}$ the indices of variables from $\Xexact$ s.t.\ $i_j < k < i_{j+1}$\;
	$\mathbb{E}_{\calT}[x] \leftarrow
        \textup{ExpectedValFrag}(i_j, i_{j+1}, k, v_j, v_{j+1})$ 
        \tcp*{Expected value of~$x$ in~$\calT$} \label{lin:compexp}
     
     $\pvol(\calT)\leftarrow \prod_{l = 0}^{m-2}  \textup{VolumeFrag}(i_l, i_{l+1},v_l,v_{l+1})$
     \tcp*{Volume of~$\calT$}
     \label{lin:compvol}
      $\mathbb{E}_{\calC}[x] \leftarrow \mathbb{E}_{\calC}[x] + \pvol(\calT)
      \times \mathbb{E}_{\calT}[x]$\tcp*{Sum exp.\ value of~$x$ weighted by $V(\calT)$}\label{lin:sumexp}
      $\pvol(\calC)\leftarrow \pvol(\calC) + \pvol(\calT)$\tcp*{Sum of total order volumes}\label{lin:sumvol}
      
}
\Return{$\frac{\mathbb{E}_{\calC}[x]}{\pvol(\calC)}$}\label{lin:last}\;
   \SetKwProg{myfunc}{Function}{}{}
   \myfunc{$\textup{ExpectedValFrag}(p, q, k, \alpha,\beta)$}{
      	\KwIn {$p<q$: indices; $k$: requested variable index; $\alpha < \beta$: exact values at $p,q$ resp.}
        \KwOut{Expected value of variable~$x_k$ in the fragment $\calC^{q-1}_{p+1}(\alpha,
        \beta)$}
      	$n\leftarrow q-p-1$ \tcp*{num.\ of variables in the fragment which are $\not\in\Xexact$}
        \Return $\frac{k-p}{n+1}\cdot(\beta-\alpha) + \alpha$ \tcp*{linear interpolation (see Section~\ref{sec:total})}}
    \myfunc{$\textup{VolumeFrag}(p,q,\alpha,\beta)$}{
      	\KwIn {$p<q$: indices; $\alpha < \beta$: exact values at $p,q$ resp.}
      	\KwOut {$\pvol(\calC_{p+1}^{q-1}(\alpha,\beta))$: volume of the total order fragment between indices $p,q$}
      	$n\leftarrow q-p-1$ \tcp*{num.\ of variables in the fragment which are $\not\in\Xexact$}
      	\Return $\frac{(\beta-\alpha)^{n}}{n!}$\tcp*{Volume of
        $[\alpha,\beta]^n$ divided by num.\ of total orders}\label{lin:fragvol}}
   \caption{Compute the expected value of a variable}
\label{alg:brute}
\end{algorithm}
Algorithm~\ref{alg:brute} presents our general scheme to compute the expected
value of a variable $x \in \calX$ under an arbitrary constraint set~$\calC$, assuming the uniform distribution on $\pw(\calC)$.

The algorithm iterates over each linear extension $\calT$ of~$\calC$, and
computes the expected value of $x$ in~$\calT$
and the overall probability of $\calT$ in $\pw(\calC)$.
{A linear extension is a total order, so $x$ is within a particular
fragment of it, namely, between the indices of two consecutive variables
with exact-value constraints, $i_j$ and $i_{j+1}$. The \emph{expected value of~$x$
in~$\calT$}, denoted by $\mathbb{E}_{\calT}[x]$, is then affected only by the
constraints and variables of this fragment, and can be computed using linear interpolation by the function
\textup{ExpectedValFrag} (line~\ref{lin:compexp}). 

Now, the final expected value of $x$
in~$\calC$ is the average of all $\mathbb{E}_{\calT}[x]$ weighted by the
probability of each linear extension $\calT$, i.e., the volume of~$\pw(\calT)$
divided by the volume of~$\pw(\calC)$. Recall that, by Lemma~\ref{lem:noties}, worlds with ties have total volume~0 and do not affect this expected value.
We compute the volume of $\calT$ as the \emph{product of volumes of its
fragments} (line~\ref{lin:compvol}). The volume of a fragment, computed by
function \textup{VolumeFrag}, is
the volume of $[\alpha,\beta]^n$, i.e., all assignments to the $n$
variables of the fragment in~$[\alpha,\beta]$, divided by the number of
orderings of these variables, to obtain the volume of one specific order 
(line~\ref{lin:fragvol}).}

The complexity of Algorithm~\ref{alg:brute} is polynomial in the number of linear extensions of~$\calC$, as
we can enumerate them in constant amortized time~\cite{PruesseR94}.
However, in the general case, there may be up to $\card{\calX}!$ linear
extensions. To obtain an upper bound in the general case, we note that we can
rescale all constraints so that all numbers are integers, and then
nondeterministically sum over the linear extensions. This yields our $\fpsp$ upper bound:

\begin{toappendix}
  \subsection{General Constraint Sets (Section~\ref{sec:general})}
\end{toappendix}

\begin{theoremrep}\label{thm:expectmember}
  Given a constraint set\/ $\calC$ over $\calX$ and $x
  \in \calX$ (resp., and a selection
  predicate~$\sigma$, and an integer~$k$), determining the expected value
  of $x$ in $\pw(\calC)$ under the uniform distribution (resp., the top-$k$ computation problem
  over $\calX$, $\calC$, and $\sigma$) is in $\fpsp$.
\end{theoremrep}

The $\fpsp$ membership 
for interpolation
does not extend to centroid computation in
general convex polytopes, which is not in $\fpsp$~\cite{lawrence1991polytope,rademacher2007approximating}.
Our algorithm thus relies on the fact that the polytope $\pw(\calC)$ is of a
specific form, defined with order {and exact-value} constraints.
The same upper bound for the top-$k$ problem immediately follows. We will show in Section~\ref{sec:complexity} that this $\fpsp$ upper bound is tight.

\begin{proof}
  Let $\calC$ be an arbitrary constraint set with order constraints and exact-value
  constraints on a variable set $\calX$, and let $n \defeq \card{\calX}$.
  Use Lemma~\ref{lem:noties} to ensure that there are no ties.
  To simplify the reasoning, we will make all values occurring in exact-value
  constraints be integers that are multiples of $(n+1)!$ as follows: let $\Delta$ be $(n+1)!$ times
  the product of the denominators of all exact-value constraints occurring in
  $\calC$, which can be computed in PTIME, and consider $\calC'$ the constraint
  set defined on $[0, \Delta]^n$ by keeping the same variables and order
  constraints, and replacing any exact-value constraint $x_i = v$ by $x_i = v
  \Delta$; the constraint set $\calC'$ is computable in PTIME from $\calC$, and the polytope
  $\pw(\calC')$ is obtained by scaling $\pw(\calC)$ by a factor of $\Delta$
  along all dimensions; hence, if we can compute the expected value of $x_i$ in
  $\calC'$ (which is the coordinate of the center of mass of $\pw(\calC')$ on
  the component corresponding to $x_i$), we can compute the expected value of
  $x_i$ in $\calC'$ by dividing by $\Delta$. So we can thus assume that
  $\pw(\calC')$ is a polytope of $[0, \Delta]^n$ where all exact-value
  constraints are integers which are multiples of $(n+1)!$.

  We use Lemma~5.2 of~\cite{abiteboul2011capturing} to argue that the volume of
  $\pw(\calC)$ can be computed in \#P. The PTIME generating Turing machine~$T$,
  given the constraint set~$\calC$, chooses nondeterministically a linear extension of $(\calX,
  {\leq_{\calC}})$, which can clearly be represented in polynomial space and checked in
  PTIME. The PTIME-computable function $g$ computes, as a rational, the volume of the polytope
  for that linear extension, and does so according to the scheme of
  Section~\ref{sec:brute}: the volume is the product of the volumes of each $\calC_1^n(\alpha,
  \beta)$, whose volume is $\frac{(\beta - \alpha)^n}{n!}$. This
  is clearly PTIME-computable, and as $\alpha$ and $\beta$ are values occurring
  in exact-value constraints, they are integers and multiples of $n!$, so the
  result is an integer, so the overall result is a product of integers, hence an
  integer. By Lemma~5.2 of~\cite{abiteboul2011capturing}, $V(\calC')$,
  which is the sum of $V(\calT)$ across all linear extensions $\calT$
  of~$\calC'$ (because there are no ties), is computable in \#P.

  We now apply the same reasoning to show that the sum, across all linear
  extensions $\calT$, of $V(\calT)$ times the expected value of
  $x_i$ in~$\calT$, is computable in \#P. Again, we use Lemma~5.2
  of~\cite{abiteboul2011capturing}, with $T$
  enumerating linear extensions, and with a function $g$ that computes the
  volume of the linear extension as above, and multiplies it by the expected
  value of $x_i$, by linear interpolation in the right $\calC_{p+1}^{q-1}$ as
  in the previous section (it is an integer, as all values of exact-value
  constraints are multiples of~$q-p\leq n+1$). So this
  concludes, as the expected value $v$ of~$x_i$ is $\frac{1}{V(\calC')}
  \sum_{\calT} V(\calT) v_{\calT}$ where $v_{\calT}$ is the expected value
  of~$x_i$ for $\calT$, and we can compute both the sum and the denominator in
\#P. Hence, the result of the division, and reducing back to the answer for the
original $\calC$, can be done in $\fpsp$.
\end{proof}

We also provide a complete example to illustrate the constructions of this section.

{\setlength{\columnsep}{1.5em}
  \setlength{\intextsep}{0em}
\begin{wrapfigure}[]{r}{0pt}
\centering
\begin{tikzpicture}[xscale=.9,yscale=1.3]
\node (x) at (0,-2) {$x$};
\node[anchor=center] (yy) at (1,-1) {$y' = \gamma$};
\node (y) at (-1,-1) {$y$};
\node (z) at (0,0) {$z$};
\draw[->] (x) -- (y.south);
\draw[->] (x) -- (yy.south);
\draw[->] (y.north) -- (z);
\draw[->] (yy.north) -- (z);
\end{tikzpicture}
\end{wrapfigure}
\subparagraph*{Full Example.}
We exemplify our scheme on
variables $\calX = \{x, y, y', z\}$ and on the constraint set
$\calC$ generated by the order constraints $x \leq y$, $y \leq z$, $x \leq
y'$, $y' \leq z$ and the exact-value constraint $y' = \gamma$ for some fixed $0 <
\gamma < 1$.
Remember that we
necessarily have $0 \leq x$ and $z \leq 1$ as well. The constraints of $\calC$
are closed under implication, so they also include $x \leq z$.
The figure shows 
the Hasse diagram of 
the partial order defined by $\calC$ on $\calX$.
Note that ties have a probability of zero in $\pw(\calC)$.

The two linear extensions of $\calC$ are $\calT_1: x \leq y \leq y' \leq z$
and $\calT_2: x \leq y' \leq y \leq z$. 
Now, $\calT_1$ is a fragmented total order, and we have
$\pw(\calT_1) = \pw_{\{x,y\}}(\calC') \times \{\gamma\} \times [\gamma, 1]$ where $\calC'$ is
defined on variables $\{x, y\}$ by $0 \leq x \leq y \leq \gamma$. We can compute
the volume of $\pw(\calT_1)$ as $\alpha_1 = \smash{\frac{\gamma^2}{2}}\times (1 -
\gamma)$.
Similarly the volume of $\pw(\calT_2)$ is $\alpha_2 = \gamma \times \smash{\frac{(1 -
\gamma)^2}{2}}$.

Let us compute the 
expected value of~$y$ for $\calC$. In $\calT_1$ its expected value is \(
\textstyle \mu_1=\mathbb{E}_{\calT_{1}}[y]=\frac{2}{3}\cdot(\gamma-0)+0 = \frac{2}{3}\gamma.
\) In $\calT_2$ its expected value is \(
\textstyle \mu_2=\mathbb{E}_{\calT_{2}}[y]=\frac{1}{3}\cdot(1-\gamma)+\gamma = \frac{1+2\gamma}{3}.
\) The overall expected value of $y$ is the average of these expected values weighted by total order probabilities (volumes fractions), namely  \(
\textstyle\mathbb{E}_{\calC}[y]=\frac{\alpha_1 \mu_1 +
	\alpha_2 \mu_2}{\alpha_1 + \alpha_2}\).}

\begin{toappendix}
\subsection{Marginal Distributions}

Beyond expected values, our results from Section~\ref{sec:brute} can serve to compute the
\emph{marginal distributions} of variables under the uniform distribution; these can be used for other estimation schemes of unknown variables (e.g., other moments of the variables, median values, etc.).

We now formalize this notion.
Letting $\calC$ be a constraint set, $x\in\calX\backslash \Xexact$ a variable
with no exact-value constraint, and $v\in[0,1]$ a value, we write
$\restr{\calC}{x=v}$ the \emph{marginalized} constraint set $\calC\cup\{x=v\}$.
{Note that, if the admissible polytope $\pw(\calC)$ is $d$-dimensional, then
	$\pw(\restr{\calC}{x=v})$ is $(d-1)$-dimensional.}
We can now define:

\begin{definition}
	\label{def:marginal}
	The \emph{marginal distribution} of an unknown variable $x\in\calX\backslash\Xexact$ is defined as the pdf $p_{x}(v)\colonequals V_{d-1}(\restr{\calC}{x=v}) / V_d(\calC)$.
\end{definition}

For unfragmented total orders, the distribution of unknown variables relates to the known problem of computing \emph{order statistics}:

\begin{definition}
	The \emph{$i$-th order statistic} for $n$ samples of a probability
	distribution $\Pr$ is the distribution $\Pr_i$ of the following process: draw $n$
	independent values according to $\Pr$, and return the $i$-th smallest value of
	the draw.
\end{definition}

\begin{proposition}[(\cite{gentle2009computational}, p.~63)]
	The distribution of the $i$-th order statistic for $n$ samples of the uniform distribution on
	$[0, 1]$ is the Beta distribution $B(i, n+1-i)$.
\end{proposition}

The connection to the marginal distribution of $x_i$ in $\calC_1^n(\alpha, \beta)$ is the following:

\begin{observation}\label{obs:n-order-statistic}
	The marginal distribution of $x_i$ within
	the admissible polytope of an unfragmented total order $\calC_1^n(\alpha, \beta)$ is the
	distribution of the $i$-th order statistic for $n$ samples of the uniform
	distribution on~$[\alpha, \beta]$.
\end{observation}

\begin{proof}
	The distribution on $n$ uniform independent samples in $[\alpha, \beta]$ can
	be described as first choosing a total order for the samples, uniformly among
	all permutations of $\{1, \ldots, n\}$ (with ties having a probability of $0$
	and thus being neglected). Then, the distribution for each total order is
	exactly that of~$x_i$ when variables are relabeled.
\end{proof}

Note that the mean of this distribution corresponds to a linear interpolation of the unknown variables between $\alpha$ and $\beta$, which we use in section~\ref{sec:total}.

We summarize these findings, as follows:
\begin{proposition}\label{prp:total-order-marginal}
	Given $\calC$, a constraint set implying a total
	order, 
	the marginal distribution 
	of any variable $x_i\in\calX$ is a polynomial of degree at most $n$ and can be computed in PTIME.
\end{proposition}

Algorithm~\ref{alg:brute} can be adapted 
to compute marginal distributions, by replacing line~\ref{lin:compexp} to
compute instead the marginal distribution of~$x$ in~$\calT$ on
variables $x_{i_j} \ldots x_{i_{j+1}}$: according to
Observation~\ref{obs:n-order-statistic}, this marginal distribution is
polynomial and defined on the range $[v_{i_j}, v_{i_{j+1}}]$. We modify
line~\ref{lin:sumexp} to compute the sum of the marginal distributions instead,
seeing them as \emph{piecewise-polynomial} functions (which are zero outside of
their range), still weighted by the volume of~$\calT$. We deduce that the
marginal distribution of~$x$ in~$\calC$ (computed by modifying
line~\ref{lin:last}) is a \deft{piecewise-polynomial} function with at most $\card{\calX}$ pieces and degree at most $\card{\calX}$.

\end{toappendix}

\section{Hardness and Approximations}\label{sec:complexity}

We next show that the intractability of Algorithm~\ref{alg:brute} in
Section~\ref{sec:brute} is probably unavoidable. We first show matching lower
bounds for interpolation and top-$k$ in Section~\ref{sec:hardness}.
We then turn in Section~\ref{sec:approx} to the problem of approximating
expected values.

\subsection{Hardness of Exact Computation}\label{sec:hardness}
\begin{toappendix}
  \subsection{Hardness of Exact Computation (Section~\ref{sec:hardness})}
  \label{apx:sec:hardness}
\end{toappendix}
We now analyze the complexity of computing an exact solution to our two main problems. We show below a new result for the hardness of top-$k$. But first, we state the lower bound for the interpolation problem, which is obtained via the geometric characterization of the problem. In previous work, centroid computation is proven to be hard for \emph{order polytopes}, namely, polytopes without exact-value
constraints, which are a particular case of our setting:

\begin{theorem}\emph{\textbf{\textsf{(\cite{rademacher2007approximating}, Theorem~1).}}}
  \label{thm:expecthard}
  Given a set\/ $\calC$ of order constraints 
  and $x \in \calX$, determining the expected value
  of $x$ in $\pw(\calC)$ under the uniform distribution is
$\fpsp$-hard.
\end{theorem}

We now show a new lower bound for top-$k$ queries: interestingly, these queries are $\fpsp$-hard
\emph{even if they do not have to return the expected values}. Recall that
$\sigma$ is the selection operator (see Section~\ref{sec:deftopk}),
which we use to compute top-$k$ among a restricted subset of variables. We can
show hardness even for top-$1$ queries, and even when $\sigma$ only selects two
variables:

\begin{toappendix}
	The following lemma provides a direct connection between the expected value of a variable to its expected rank, for constraint sets that consist only of order constraints. By the hardness of computing expected rank in partial orders~\cite{brightwell1991counting} this provides an alternative proof to the hardness of computing expected value (Theorem~\ref{thm:expecthard}).
	We use this connection to obtain an upper bound on the denominator of the expected value, when we reduce to the problem of top-$k$ computation below.
	
\begin{lemma}\label{lem:ranktoexpect}
	Given $\calX$ and a set of order constraints $\calC$ over $\calX$, if the expected rank of $x\in\calX$ is $r$ then its expected value is \(\textstyle\frac{r}{n+1}\).
\end{lemma}
\begin{proof}
	Let $N$ be the number of linear extensions of $\calC$.
	Since $\calC$ only consists of order constraints, by symmetry, the probability of each of its linear extensions is identical, namely $1/N$. Denote by $r_i$ the rank of $x$ in the linear extension $\calT_i$. Then $r= \sum_{i=1}^{N}\frac{1}{N} r_i$. By Algorithm~\ref{alg:brute}, the expected value of $x$ relative to $\calT_i$ is $\frac{r_i}{n+1}$, and thus the overall expected value of $x$ relative to $\calC$ is $\mathbb{E}_{\calC}[x]=\sum_{i=1}^{N}\frac{1}{N} \frac{r_i}{n+1}=\frac{r}{n+1}$.
\end{proof}
\end{toappendix}

\begin{theoremrep}\label{thm:topkhard}
Given a constraint set $\calC$ over~$\calX$, a selection
  predicate~$\sigma$, and an integer~$k$, the top-$k$ computation problem
  over $\calX$, $\calC$ and $\sigma$ is $\fpsp$-hard  even if $k$ is
  fixed to be~$1$, $|\calX_\sigma|$~is~$2$, and the top-$k$ answer does not include the expected value of the variables.
\end{theoremrep}

\begin{proof}
  We will perform a reduction from the interpolation problem (i.e., expected value
  computation) on sets of order constraints, which is $\fpsp$-hard by Theorem~\ref{thm:expecthard}. Write $n
  \defeq \card{\calX}$. Assume using Lemma~\ref{lem:noties} that the
  probability of ties in~$\pw(\calC)$ is zero.

  We first observe, using Lemma~\ref{lem:ranktoexpect}, that
  the expected value $v$ of $x_i$ can be written as $\frac{1}{N (n+1)}
  \sum_{\calT} r_{\calT}$, where the sum is over all the linear
  extensions $\calT$ of $\calC$, $r_{\calT} \in \{1, \ldots, n\}$ is the rank
  of $x_i$ in the linear extension $\calT$, and $N \leq n!$ is the number of linear
  extensions. This implies that $v$ can be written as a rational $p/q$
  with $0 \leq p \leq q$ and $0 \leq q \leq M$, where we write $M \defeq (n+1)!$.

  We determine this fraction $p/q$ using the algorithm
  of~\cite{papadimitriou1979efficient}, that proceeds by making queries of the
  form ``is $p/q \leq p'/q'$'' with $0 \leq p', q' \leq M$, and runs in time
  logarithmic in the value $M$, so polynomial in the input $\calC$. To do so, we must
  describe how to decide in PTIME whether $v \leq p'/q'$ for $0 \leq p', q' \leq
  M$, using an oracle for the top-$1$ computation problem that does not return the expected values.

  Fix $v' = p'/q'$ the query value and let $v = p/q$ be the unknown target value, the
  expected value of $x_i$. We illustrate how to decide whether $v \leq v'$. The
  general idea is to add a variable with exact-value constraint to $v'$ and
  compute the top-$1$ between $x_i$ and the new variable, but we need a slightly
  more complicated scheme because the top-$1$ answer variable can be arbitrary in the
  case where $v = v'$ (i.e., we have a tie in computing the top-$1$).
  Let $\epsilon \defeq 1/(2(M^2+1))$, which is computable in PTIME in the value
  of $n$ (so in PTIME in the size of the input $\calC$).
  Construct $\calC'$ (resp., $\calC'_-$, $\calC'_+$) by adding an exact-value
  constraint $x' = v'$ (resp., $x'_- = v' - \epsilon$, $x'_+ = v' + \epsilon$)
  for a fresh variable $x'$ (resp., $x'_-$, $x'_+$).
  Now use the oracle for $\calC'$ (resp., $\calC'_-$, $\calC'_+$) and the
  selection predicate that selects $x_i$ and $x'$ (resp., $x'_-$, $x'_+$), taking $k = 1$ in all
  cases. The additional variables do not affect the expected value of $x_i$ in
  $\calC'$, $\calC'_+$, and $\calC'_-$, so it is also $v$ in them.
  Further, we know that $v = p/q$, $v' = p'/q'$, with $0 \leq p, q, p', q' \leq M$, hence either $v = v'$, or
  $\left|v - v'\right| \geq \frac{1}{M^2}$. Hence, letting $S = \{v',
  v'-\epsilon,v'+\epsilon\}$, there are three
  possibilities: $v = v'$, $v < v''$ for all $v'' \in S$, or $v > v''$ for all
  $v'' \in S$. Thus, if the top-$1$ variable in all oracle calls is always $x_i$
  (resp., never $x_i$), then we are sure that $v > v'$ (resp., $v < v'$). If
  some oracle calls return $x_i$ but not all of them, we are sure that $v = v'$.
  Hence, we can find out in PTIME using the oracle whether $v \leq v'$.
  This concludes the proof, as we then have an overall PTIME reduction from the
  $\fpsp$-hard problem of interpolation (i.e., expected value computation) to the top-$1$ computation
  problem, showing that the latter is also $\fpsp$-hard.
\end{proof}

\begin{toappendix}
  We conclude with a proof of the result variant stated informally after
  Theorem~\ref{thm:topkhard}:

  \begin{theorem}
    Given a constraint set $\calC$ over~$\calX$ and an integer~$k$, the top-$k$
    computation problem over $\calX$ and $\calC$
    is $\fpsp$-hard even for the trivial selection predicate $\sigma$ (with
    $\calX_\sigma = \calX$) and even when the top-$k$ answer does not include
    the expected value of the variables.
  \end{theorem}

  \begin{proof}
    We prove the result exactly like Theorem~\ref{thm:topkhard}, except that,
    instead of computing the top-$1$ between two variables $x_i$ and some $x'
    \in \Xexact$,
    we compute the top-$k$ over all variables for different values of~$k$. Note that we can assume that the expected value $v$ of~$x_i$ is not equal
    to the exact value $v'$ of~$x'$ (otherwise, we can resolve it using the technique of Theorem~\ref{thm:topkhard}). Under this assumption, we have $v < v'$ iff the smallest $k$
    such that $x$ occurs in the top-$k$ is less than the smallest $k'$ such that
    $x'$ occurs in the top-$k'$, and likewise $v > v'$ whenever $k > k'$. $k$ and $k'$ can be found e.g.\ using binary search.
  \end{proof}
\end{toappendix}

\begin{proofsketch}
To prove hardness in this case, we reduce from interpolation. We show that a
top-1 computation oracle can be used as a comparison oracle to compare the
expected value of a variable $x$ to any other rational value $\alpha$, by adding
  a fresh element $x'$ with an exact-value constraint to~$\alpha$ {and
  using $\sigma$ to compute the top-1 among $\{x,x'\}$}. What is more
technical is to show that, given such a comparison oracle, we can perform the
reduction and determine \emph{exactly} the expected value $v$ of $x$ (a rational
number) using only a \emph{polynomial} number of comparisons to other rationals.
  This follows from a bound on the denominator of~$v$, and by applying the
  rational number identification scheme of~\cite{papadimitriou1979efficient}.
  See the Appendix.
\end{proofsketch}

{In settings where we do not have a selection operator (i.e.,
$\calX_\sigma = \calX$), we can similarly show the
hardness of top-$k$ (rather than top-$1$). See Appendix~\ref{apx:sec:hardness} for details.}

\subsection{Complexity of Approximate Computation}\label{sec:approx}
\begin{toappendix}
  \subsection{Complexity of Approximate Computation (Section~\ref{sec:approx})}
\end{toappendix}

In light of the previous hardness results, we now review approximation
algorithms, again via the geometric characterization of our setting. In
Section~\ref{sec:tractable}, we will show a novel exact solution in PTIME for specific cases.

The \emph{interpolation} problem can be shown to admit a
fully polynomial-time randomized approximation scheme (FPRAS). This result
follows from existing work \cite{kannan1997random,bertsimas2004solving}, using a tractable almost uniform sampling scheme
for convex bodies.

\begin{propositionrep}\emph{\textbf{\textsf{(\cite{kannan1997random}, Algorithm~5.8).}}}
  \label{prp:fpras}
  Let $\calC$ be a set of constraints with variable set $\calX$ and
  $x\in\calX$. There is an
  FPRAS that determines an estimate $\hat{\mathbb{E}}_\calC[x]$ of the
  expected value~$\mathbb{E}_\calC[x]$ of $x$ in $\pw(\calC)$ under the
  uniform distribution.
\end{propositionrep}

\begin{toappendix}
  For clarity and completeness, we give a self-contained proof of this result,
  using the almost uniform sampling scheme of~\cite{kannan1997random}:
\end{toappendix}

\begin{proof}
We use a result by Kannan, Lov\'asz, and
Simonovits~(Theorem~2.2 of~\cite{kannan1997random})
that shows that sampling a point \emph{almost uniformly} from a convex
body in dimension $n$
can be done by
$\tilde O(n^5)$ calls to an oracle deciding membership in the body, and an
additional factor of $\tilde O(n^2)$ arithmetic operations (here, $\tilde
O(\cdot)$ is the \emph{soft-O} notation, that ignores polylogarithmic factors). 
``Almost uniformly'' means that the \emph{total variation distance} between the
uniform distribution and the actual distribution realized by the sampling
is less than any fixed number $\epsilon'>0$. The dependency of the
  running time in $\epsilon'$ is logarithmic, hence hidden in $\tilde O(\cdot)$.

Let $\epsilon,\delta>0$ be two reals. For some number $N$ that we define
further, we first apply $N$ times consecutively the sampling algorithm
of~\cite{kannan1997random} with $\epsilon'\defeq\frac\epsilon 4$ to
  obtain a set $P = \{p_1, \ldots, p_N\}$ of~$N$ independent samples.
The algorithm uses an oracle for membership to the polytope, which in our
setting can be determined in time $O(|\calC|)$; hence, as
the dimension
of $\pw(\calC)$ is bounded by $|\calX|$, by the complexity analysis
of~\cite{kannan1997random}, this process has a running time
of $O(N\times|\calX|^7\times|\calC|)$.

We then consider the projection of every point $p_i \in P$ to the dimension
defined by the variable $x$, giving $N$ values $v_1,\dots, v_N$. We
then compute the estimate $\hat{\mathbb{E}}_\calC[x]$ as $\frac 1
N\sum_{i=1}^N v_i$. Clearly this has no impact on the running time. It only
remains to show that this estimate satisfies the given bounds, i.e., that 
  \(
  \left|\hat{\mathbb{E}}_\calC[x]-\mathbb{E}_\calC[x]\right|\leq\epsilon\)
  with probability at least $1-\delta$.

Let $f:\pw(C)\to[0;1]$ be
the probability density function of the distribution realized by the
algorithm of~\cite{kannan1997random}, according to which the independent samples
were drawn.
By definition of the total variation distance, we know that
$\frac 1 2\int_{\pw(\calC)} \left|\frac{1}{\pvol(\calC)}-f(p)\right|\d p\leq\frac{\epsilon}{4}$.
If we denote by $\mu_v$ the average of the (independent)
$x$-value samples $v_1,\dots, v_N$, and we denote by $f_{x=t}$ the function obtained
from~$f$ by setting to~$t$ the coordinate corresponding to~$x$, we have:
\begin{align*}
  \left|\mu_v-\mathbb{E}_\calC[x]\right|&=\left|\int_0^1
  t\int_{\pw(\calC_{x=t})}
  f_{x=t}(p)\d p\d t \quad - \quad \int_0^1 t\int_{\pw(\calC_{x=t})}
  \frac{1}{\pvol(\calC)}\d p \d t\right|
  \\&=
  \left|\int_0^1 t\int_{\pw(\calC_{x=t})}
  \left(f_{x=t}(p)-\frac{1}{\pvol(\calC)}\right)\d p\d t\right|
  \\&\leq
  \int_0^1 1\times\int_{\pw(\calC_{x=t})}
  \left|\frac{1}{\pvol(\calC)}-f_{x=t}(p)\right|\d p\d t
  \\&= \int_{\pw(\calC)}\left|\frac{1}{\pvol(\calC)}-f(p)\right|\d p\leq\frac\epsilon 2.
\end{align*}

Now, as the $v_i$ are independent and identically distributed, we obtain
by Hoeffding's inequality~\cite{Hoeffding63}:
\[
\Pr\left(\left|\hat{\mathbb{E}}_\calC[x]-\mu_v\right|\geq\frac\epsilon
2\right)
\leq
2\exp\left(\frac {-\epsilon^2 N}{2}\right).
\]

Therefore, setting $N\defeq \frac{2\ln\left(\frac{2}{\delta}\right)}{\epsilon^2}$, we have:
\[
  \Pr\left(\Big|\hat{\mathbb{E}}_\calC[x]-\mathbb{E}_\calC[x]\Big|\leq\epsilon\right)\geq
  \Pr\left(\Big|\hat{\mathbb{E}}_\calC[x]-\mu_v\Big|\leq\frac\epsilon
  2 \land \Big|\mu_v-\mathbb{E}_\calC[x]\Big|\leq\frac\epsilon
2\right)\geq1-\delta. \qedhere
\]
\end{proof}

This result is mostly of theoretical interest, as
the polynomial
is in~$\card{\calX}^7$ (see \cite{bertsimas2004solving}, Table~1),
but
recent improved sampling algorithms~\cite{LovaszV06}
may ultimately yield a practical
approximate interpolation technique for general constraint sets (see~\cite{LovaszD12,GeM15}).

For completeness, we mention two natural ways to
define randomized approximations for 
\emph{top-$k$ computation}:

\begin{itemize}
  \item We can define the \emph{approximate top-$k$} as an ordered list
    of $k$ items whose expected value does not differ by more than some
    $\epsilon>0$ from that of the item in the actual top-$k$ at the same
    rank. An FPRAS for this definition of approximate top-$k$
    can be
    obtained from that of
    Proposition~\ref{prp:fpras}.

  \item \label{itm:topapprox} It is highly unlikely that there exists a PTIME
    algorithm to return the \emph{actual top-$k$} with high probability, even without
    requiring it to return the expected values. Indeed, such an algorithm would be in the BPP
  (bounded-error probabilistic time)
    complexity class; yet it follows from Theorem~\ref{thm:topkhard} above that
    deciding whether a set of variables is the top-$k$ is NP-hard, so the
    existence of the algorithm would entail that $\text{NP} \subseteq \text{BPP}$. 
\end{itemize}

\section{Tractable Cases}\label{sec:tractable}
Given the hardness results 
in the previous section
and the impracticality of approximation,
we now study whether
\emph{exact} interpolation and top-$k$ computation can be tractable
on restricted classes of constraint sets.
We consider tree-shaped constraints (defined formally below) and generalizations
thereof: they are relevant for practical applications (e.g., classifying items
into tree- or forest-shaped taxonomies), and we will show that our problems are
tractable on them. 
We start by a splitting lemma 
to decompose constraint sets into ``independent'' subsets of variables,
and then define and study our tractable class.

\subsection{Splitting Lemma}\label{sec:splitting}
\begin{toappendix}
  \subsection{Splitting Lemma (Section~\ref{sec:splitting})}
  \label{apx:splitting}
\end{toappendix}
We will formalize the cases in which the valuations of two variables in $\calX$
are probabilistically dependent (the variables \emph{influence} each other),
according to $\calC$. This, in turn, will enable us to define \emph{independent
subsets of the variables} and thus \emph{independent subsets of the constraints}
over these variables. This abstract result will generalize the notion of fragments from total orders (see Section~\ref{sec:total}) to general constraint sets.
In what follows, we use $x_i \prec x_j$ to denote the \emph{covering relation}
of the partial order~$\leq$, i.e., $x_i \leq x_j$ is in $\calC$ but there exists no $x_k \notin
\{x_i, x_j\}$ such that $x_i \leq x_k$ and $x_k \leq x_j$ are in $\calC$.

\begin{definition}
  {We define the \emph{influence relation} $x \leftrightarrow y$
  between variables of $\calX \backslash \Xexact$ as the equivalence relation
  obtained by the symmetric, reflexive, and transitive closure of the $\prec$
  relation on $\calX \backslash \Xexact$.}
  
  {The \emph{uninfluenced classes} of $\calX$ under $\calC$ is the
  partition of $\calX\backslash\Xexact$ as the subsets $\mathsf{X}_1, \dots,
  \mathsf{X}_m$ given by the equivalence classes of the influence relation.}
  
  {The \emph{uninfluence decomposition of $\calC$} is the collection
  of constraint sets $\calC_1,\dots,\calC_m$ of $\calC$ where each $\calC_i$
  has as variables $\mathsf{X}_i \sqcup \Xexact$ and contains all exact-value
  constraints of~$\calC$ and all order constraints between variables of
  $\mathsf{X}_i \sqcup \Xexact$.}
\end{definition}

We assume w.l.o.g.\ that $m>0$, i.e., there are unknown variables in
$\calX\backslash\Xexact$; otherwise the uninfluence decomposition is meaningless but any analysis is trivial.
Intuitively, two unknown variables $x,x'$ are in different
uninfluenced classes if in \emph{every} linear extension there is
\emph{some} variable from $\Xexact$ between them, or if they belong to disconnected (and thus incomparable) parts of the partial order. 
In particular, uninfluenced classes
correspond to the fragments of a total order: this is used in
Section~\ref{sec:total}.
The uninfluence
decomposition captures only constraints between \emph{variables that influence
each other}, and constraints that can \emph{bound the range of a variable} by
making it comparable to variables from~$\Xexact$.
We formally prove the independence of $\calC_1,\dots,\calC_m$ via
possible-world semantics: every possible world of $\calC$ can be decomposed to possible worlds of  $\calC_1,\dots,\calC_m$, and vice versa. 

\begin{lemmarep}\label{lem:splittingb}
  Let $\calC_1,\dots,\calC_m$ be the uninfluence decomposition of $\calC$. 
  There exists a bijective correspondence between $\pw(\calC)$ and $\pw(\calC_1) \times \cdots \times \pw(\calC_m)$.
\end{lemmarep}

\begin{proof}
	Define the function  $M:\pw(\calC)\rightarrow \pw_{\mathsf{X}_{1}\sqcup\Xexact}(\calC_1) \times \cdots \times
	  \pw_{\mathsf{X}_{m}\sqcup\Xexact}(\calC_m)$,
	that maps each possible world $w$ of $\calC$ to a
        tuple $\langle w_1, \ldots, w_m \rangle$ of possible worlds of $\calC_1,
        \ldots, \calC_m$ defined as follows: for each variable $x \in \calX$,
        letting $v$ its value in $\calC$, we give to~$x$ the value $v$ 
        in all possible worlds of $w_1, \ldots, w_m$ where $x$
        appears.
        Recall that we assumed w.l.o.g.\ that $m>0$, so each variable of~$\calX$
        must occur in at least one of~$w_1, \ldots, w_m$.
	
  It is immediate that any $w \in \pw(\calC)$ yields a tuple of
  possible worlds of $\calC_1, \ldots, \calC_m$ by this definition of $M$, since they are subsets of $\calC$. 
  
  Conversely, consider any tuple
  of possible worlds $\langle w_1,\dots,w_m\rangle\in\calC_1, \ldots, \calC_m$. Note that each $x_i\in\Xexact$ is always consistently mapped to the same value as each $\calC_i$ contains every exact-value constraint. Every other variable $x_i\in\mathsf{X}_j$ appears only in $w_j\in\pw_{\mathsf{X}_{j}\cup\Xexact}(\calC_j)$. Thus $w=M^{-1}(\langle w_1,\dots,w_m\rangle)$ is well-defined.
  Let us show that $w\in\pw(\calC)$.

  First, we observe that exact-value constraints are necessarily respected,
  because each~$\calC_i$ contains all such constraints. Second, let us show that order constraints are
  respected. Consider an
  order constraint of the form $x \leq x'$ in~$\calC$. Clearly, if $x$ and~$x'$
  belong to the same uninfluenced class~$\mathsf{X}_i$ (or if at least one of them has an exact-value constraint), the constraint is
  reflected in~$\calC_i$, so it must be respected in~$w$. Hence, we focus
  on the case where
  $x\in\mathsf{X}_i$ and $x''\in\mathsf{X}_j$ with $i\neq j$.
  Now,
  as $x \leq x'$, there exists a path $\calP$ of the form $x = x_1 \prec \cdots \prec x_n = x'$,
  but as they are not in the same uninfluenced class there must be a variable
  of~$\Xexact$ in the
  sequence. Let $x_k$ be this variable with exact-value $v$. By definition of the uninfluence decomposition, $x\leq x_k$ and $x_k=v$ are in $\calC_i$; similarly, $x_k\leq x'$ and $x_k=v$ are in $\calC_j$. Thus $x\leq x'$ must be respected in $w$ overall.
\end{proof}

\begin{example}
  \label{exa:influence}
  Let $\calX$ be $\{x, y, y', z, w\}$, and let $\calC$ be defined
  by $y' = 0.5$ and $x \leq y \leq y' \leq z$. The uninfluence classes are
  $\mathsf{X}_1 = \{x, y\}$, $\mathsf{X}_2 = \{z\}$, and $\mathsf{X}_3 = \{w\}$.
  The uninfluence decomposition thus consists of $\calC_1$, with variables
  $\mathsf{X}_1 \sqcup \{y'\}$, and constraints $x \leq y \leq y'$ and $y' =
  0.5$; $\calC_2$, with variables $\mathsf{X}_2 \sqcup \{y'\}$, and constraints
  $y' \leq z$ and $y' = 0.5$; and $\calC_3$, with variables $\mathsf{X}_3 \sqcup
  \{y'\}$, and constraint $y' = 0.5$.

\end{example}
We next use this independence property
to analyse restricted classes of constraint sets.

\subsection{Tree-Shaped Constraints}\label{sec:trees}
\begin{toappendix}
  \subsection{Tree-Shaped Constraints (Section~\ref{sec:trees})}
  \label{apx:sec:trees}
\end{toappendix}
We define the first restricted class of constraints that we consider:
\emph{tree-shaped} constraints. 
Recall that a \deft{Hasse diagram} is a representation of a partial order as a
directed acyclic graph, whose nodes correspond to $\calX$ and where there is an
edge $(x, y)$ if $x \prec y$. An example of such a diagram is the one used in
Section~\ref{sec:general}.

\begin{definition}
  \label{def:treeshaped}
  A constraint set $\calC$ over $\calX$
  is \deft{tree-shaped} if the probability of ties is zero,
  the Hasse diagram of the partial order induced on $\calX$ by $\calC$ is a directed tree, the root has
  exactly one child, and exactly the root and leaves are in $\Xexact$.
  Thus, $\calC$ imposes a global minimal value, and maximal values at each leaf,
  and no other exact-value constraint.

  We call~$\calC$
  \deft{reverse-tree-shaped} if the reverse
  of the Hasse diagram (obtained by reversing the direction of the edges)
  is tree-shaped.
\end{definition}

{Tree-shaped constraints are often encountered in practice, in particular in the
context of product taxonomies.
Indeed, while our example from Figure~\ref{fig:catalog} is a DAG,
many real-life taxonomies are trees: in particular, the Google Product
Taxonomy~\cite{googleproduct} and ACM CCS~\cite{acmccs}.}

We now show that for a tree-shaped constraint set $\calC$, unlike the general
case, we can tractably compute 
exact expressions of the 
expected values of
variables.
In the next two results, we assume
arithmetic operations on rationals to have
unit cost, e.g., they are performed up to a
fixed numerical precision. 
Otherwise, the complexities remain
polynomial but the degrees may be larger.
We first show:

\begin{theoremrep}\label{thm:voltree}
  For any tree-shaped constraint set\/ $\calC$ over\/ $\calX$, we can
  compute its volume $V(\calC)$ in time $O(\card{\calX}^2)$.
\end{theoremrep}

\begin{proofsketch}
We process the tree bottom-up, propagating a piecewise-polynomial function
expressing the volume of the subpolytope on the subtree rooted at each node as a
function of the value of the parent node: we compute it using
Lemma~\ref{lem:splittingb} from the child nodes.
\end{proofsketch}

\begin{proof}
  Let $T$ be the tree with vertex set $\calX$ which is the Hasse diagram of the
  order constraints imposed by~$\calC$.
  For any variable $x \in \calX$ that has no exact-value constraint (so it is
  not the root of~$T$ or a leaf of~$T$), let
  $\calC_x$ be the constraint set obtained as a subset of~$\calC$ by
  keeping only constraints between $x$ and its descendants in~$T$, as well
  as between $x$ and its parent. For $v \in [0, 1]$, we call $V_x(v)$ the
  $d$-volume
  of $\pw(\calC_x \cup \{x' = v\})$ where $x'$ is the parent of~$x$ and
  $d$ is the dimension of $\pw(\calC_x)$.
  In other words, $V_x(v)$ is the $d$-volume of the admissible polytope for the
  subtree $\restr{T}{x}$
  of~$T$
  rooted at $x$,
  as a function of the minimum value on $x$ imposed by the exact-value constraint
  on the parent of~$x$. It is clear that, letting $x_\r'$ be the one child of the
  root $x_\r$ of~$T$, we have $V(\calC) = V_{x_\r'}(v_\r)$, where $v_\r$ is the exact value
  imposed on~$x_\r$.

  We show by induction on~$T$ that, for any node~$x$ of $T$, letting $m_x$ be the
  minimum exact-value among all leaves that are descendants of $x$, 
  the function $V_x$ is zero in the
  interval $[m_x, 1]$ and can be expressed in $[0, m_x]$ as a polynomial whose
  degree is at most the number of nodes in $\restr{T}{x}$, written
  $\cardb{\restr{T}{x}}$.
  Since the probability of ties is~$0$, we have $m_{x}>0$ for all~$x$.

  The base
  case is for a node $x$ of $T$ which has only leaves as children; in this
  case it is clear that $V_x(v)$ is $m_x - v$ for $v \in [0, m_x]$, and is zero
  otherwise.
  For the inductive case, let $x$ be a variable. It is clear that
  $V_x(v)$ is $0$ for $v \in [m_x, 1]$. Otherwise, let $v' \in [0, m_x]$ be the
  value of the parent $x'$ of $x$. For every value $v' \leq v \leq m_x$ of $x$,
  consider the constraint set $\calC_{x, v, v'} = \calC_x \cup \{x' =
  v', x = v\}$. By Lemma~\ref{lem:splittingb}, we have $V(\calC_{x, v, v'}) =
  \prod_i V_{x_i}(v)$ where $x_1, \ldots, x_l$ are the children of $x$. Hence,
  by definition of the volume, we know that $V_x(v') = \int_{v'}^{m_x}
  \prod_i V_{x_i}(v) \d v$. Now, we use the induction hypothesis to deduce that
  $V_{x_i}(v)$, for all $i$, in the interval $[0, m_x]$, is a polynomial whose
  degree is at most~$\cardb{\restr{T}{x_i}}$. Hence, as
  the product of polynomials is a polynomial whose degree is the sum of the
  input polynomials, and integrating a polynomial yields a polynomial whose
  degree is one plus that of the input polynomial, $V_x$ in the interval $[0,
  m_x]$ is a polynomial whose degree is at most $\cardb{\restr{T}{x}}$.

  Hence, we have proved the claim by induction, and we use it to determine
  $V(\calC)$ as explained in the first paragraph.

  We now prove that the computation is quadratic. We first assume that the tree
  $T$ is binary. We show by induction that
  there exists a constant $\alpha\geq 0$ such that the computation of the
  polynomial $V_{x_i}$ in expanded form has cost less than $\alpha n_i^2$, where
  $n_i$ is $\cardb{\restr{T}{x_i}}$. The claim
  is clearly true for nodes where all children are leaves, because the cost is
  linear in the number of child nodes as long as $\alpha$ is at
  least the number the number of operations per node~$\alpha_0$.
  For the induction step, if $x_i$ is an
  internal node, let $x_p$ and $x_q$ be the two children. By induction
  hypothesis, computing $V_{x_p}$ and $V_{x_q}$ in expanded form has cost $\leq \alpha (n_p^2 +
  n_q^2)$. Remembering that arithmetic operations on rationals are assumed to
  take unit time, computing the product of $V_{x_p}$ and $V_{x_q}$ in expanded form has
  cost linear in the product of the degrees\footnote{We could compute the
    product of the polynomials more efficiently using an FFT, but this would
  not improve the overall complexity bound.} of $V_{x_p}$ and $V_{x_q}$ which are
  less than $n_p$ and $n_q$, so the cost of computing the product is
  $\leq \alpha_1 n_p
  n_q$ for some constant $\alpha_1$. Integrating has cost linear in the degree of the resulting polynomial,
  that is, $n_p + n_q$. So the total cost of computing $V_{x_i}$
  is $\leq \alpha
  (n_p^2 + n_q^2) + \alpha_1 n_p n_q + \alpha_2(n_p + n_q) +\alpha_3$ for some
  constants $\alpha_2$, $\alpha_3$. Now, as $n_q = n_i - n_p -1$,
  computing $V_{x_i}$ costs less than:
  \begin{align*}
     &\alpha n_i^2 +2\alpha n_p^2+\alpha
   -2\alpha n_i n_p-2\alpha n_i +2\alpha n_p+\alpha_1
  n_pn_q+\alpha_2 n_i-\alpha_2+\alpha_3\\
  ={}&\alpha
  n_i^2+(\alpha_1-2\alpha)n_pn_q+(\alpha_2-2\alpha)n_i+\alpha-\alpha_2+\alpha_3
 \end{align*}
 As long as $\alpha$ is set to be $\geq \max(\frac{\alpha_1}{2},
 \frac{\alpha_2}{2})$, the second and third terms are negative,
 which means (since $n_pn_q$ and $n_i$ are both $\geq 1$) that $V_{x_i}$
 costs less than:
 \begin{align*}
  &\alpha n_i^2+\alpha_1-2\alpha+\alpha_2-2\alpha+\alpha-\alpha_2+\alpha_3
  \\={}&\alpha n_i^2-3\alpha +\alpha_1+\alpha_3\leq \alpha n_i^2
\end{align*}
if $\alpha\geq\frac{\alpha_1+\alpha_3}{3}$. This concludes the induction
case, by setting $\alpha$ to any arbitrary value which is greater than
$\max\left(\alpha_0,\frac{\alpha_1}{2},\frac{\alpha_2}{2},\frac{\alpha_1+\alpha_3}{3}\right)$.
  Hence the claim is
  proven if $T$ is binary.

  If $T$ is not binary, we use the associativity of product to make~$T$
  binary, by adding virtual nodes that represent the computation of the
  product. In so doing, the size of $T$ increases only by a constant
  multiplicative factor (recall that the number of internal nodes in a
  full binary tree is one less than the number of leaves, meaning that
  the total number of nodes in a binary expansion of a $n$-ary product is
  less than twice the number of operands of the product).
  So the claim also holds for arbitrary $T$.
\end{proof}

See Appendix~\ref{apx:sec:trees} for the complete proof. This result can be applied to prove the tractability of computing
the marginal distribution of any variable~$x\in\calX\backslash\Xexact$ in a
tree-shaped constraint set, which is defined as the pdf
$p_{x}(v)\colonequals V_{d-1}(\calC\cup\{x=v\}) / V_d(\calC)$, where $d$
is the dimension of $\pw(\calC)$:

\begin{theoremrep}\label{thm:margtree}
  For any tree-shaped constraint set $\calC$ on variable set $\calX$, for any
  variable $x\in\calX\backslash\Xexact$, the marginal distribution for $x$ is piecewise polynomial
  and can be computed in time $O(\card{\Xexact} \times \card{\calX}^2)$.
\end{theoremrep}

\begin{proofsketch}
  We proceed similarly to the proof of Theorem~\ref{thm:voltree} but with two
  functions: one for $x$ and its descendants, and one for all other nodes. The
  additional~$\card{\Xexact}$ factor is because the second function
  depends on how the value given to~$x$ compares to
  the leaves.
\end{proofsketch}

\begin{proof}
  Recall that $\calC_{|x=v}$ is $\calC$ plus the exact-value constraint $x = v$. For any
  variable $x'$, we let $m_{x'}$ be
  the minimum exact-value among all leaves reachable from $x'$.
  By
  definition, the marginal distribution for~$x$ is
  $v \mapsto \frac{1}{V(\calC)} V(\calC_{|x=v})$.
  We have seen in Theorem~\ref{thm:voltree} that
  $\frac{1}{V(\calC)}$ can be computed in quadratic time; we now focus on the
  function $V(\calC_{|x=v})$.

  By Lemma~\ref{lem:splittingb}, letting $x_1, \ldots, x_k$ be the children of
  $x$, $D_1, \ldots, D_k$ be their descendants (the $x_i$ included), and
  $D$ be all variables except $x$ and its descendants, that is, $D \colonequals
  \calX \backslash (\{x\} \cup \bigsqcup_i D_i)$, we can express $V(\calC_{|x=v})$ as
  $V'_x(v) \times \prod_i V_{x_i}(v)$, where $V_{x_i}$ is as in the proof of
  Theorem~\ref{thm:voltree}, and $V'_x(v)$ is the volume of the constraint set
  $\calC'_{x,v}$
  over~$D$
  obtained by keeping all constraints in $\calC$ about variables of
  $D$, plus the exact-value constraint $x = v$. Indeed, the uninfluenced classes
  of~$\calC_{|x=v}$ are clearly $D \backslash \Xexact$, $D_1 \backslash \Xexact,
  \ldots, D_k \backslash \Xexact$. We denote by $\calX'$ the variables of
  $\calC'_{x,v}$.

  We know by the proof of Theorem~\ref{thm:voltree} that $V_{x_i}$, in the interval $[0,
  m_x]$, is a polynomial whose degree is at most $\cardb{\restr{T}{x_i}}$, and
  that it can be computed in $\oofb{\cardb{\restr{T}{x_i}}^2}$. Hence, the
  product of the $V_{x_i}(v)$ can be computed in quadratic time
  in~$\card{\restr{T}{x}}$ overall (as in the proof of Theorem~\ref{thm:voltree})
  and it has linear degree. We thus focus on $\calC'_{x,v}$, for which it
  suffices to show that $V(\calC'_{x,v})$ is a piecewise-polynomial function
  with at most $\card{\Xexact}$
  pieces, each piece having a linear degree and being computable in quadratic
  time in $\calX'$.
  Indeed, this suffices to justify that computing the product of~$V(\calC'_{x,v})$ with $\prod_i
  V_{x_i}(v)$, and integrating to obtain the marginal
  distribution, can be done in time $O(\card{\Xexact} \times \card{\calX}^2)$,
  and that the result is indeed piecewise polynomial.

  For any node $x_i$ of~$D$ with no exact-value constraint, we let $V_{x_i,x}'(v, v')$ be
  the volume of the constraint set obtained by restricting $\calC'_{x,v'}$ to the
  descendants of the parent $x_i'$ of $x_i$ and adding the exact-value constraint
  $x_i' = v$. We let $(v_1,
  \ldots, v_q)$ be the values occurring in exact-value constraints
  in $\calC$, in increasing order, so that $q \leq \card{\Xexact}$.
  We show by induction on $D$ the following
  claim: for any $1 \leq i < q$, for any variable $x_i$ in~$D$ with no exact-value
  constraint, in the intervals $v \in [0, m_{x_i}]$ and $v' \in [v_i,
  v_{i+1}]$, $V_{x_i}'(v, v')$ can be expressed as $P(v) + v'P'(v)$, where $P$
  and $P'$ are polynomials of degree at most $\cardb{\restr{T}{x_i} \cap D}$ and
  can be computed in quadratic time in~$\cardb{\restr{T}{x_i} \cap D}$.

  The proof is the same as in Theorem~\ref{thm:voltree}: for the base case where
  all children of $x_i$ have exact-value constraints,
  $V_{x_i}(v, v')$ is either $m_{x_i} - v$ if $x$ is not reachable from $x_i$ or
  $v_i \geq m_x$, or $v' - v$ otherwise. For the inductive case, we do the same
  argument as before, noting that, clearly, taking the product of the $V_{\cdot,
  \cdot}'(v, v')$ among the children of $x_i$, the variable $v'$ occurs in at most one of
  them, namely the one from which $x$ is reachable.
  We conclude that $V'_x(v)$ is indeed a piecewise polynomial function
  with at most $\card{\Xexact}$ many pieces, all of which have a linear degree, by evaluating $V_{x'',
  x}'(v'', v)$, where $x''$ is the one child of the root of $T$ and $v''$ is
  the value to which it has an exact-value constraint. The overall
  computation time is then in $O(\card{\Xexact} \times \card{\calX}^2)$.
\end{proof}

\medskip

We last deduce that our results for tree-shaped constraints extend to a more
general tractable case: constraint sets $\calC$ whose uninfluence decomposition
$\calC_1, \ldots, \calC_m$ is such that every $\calC_i$ is
(reverse-)tree-shaped. By Lemma~\ref{lem:splittingb}, each $\calC_i$ (and its
variables) can be considered independently, and reverse-tree-shaped trees can be
easily transformed into tree-shaped ones. Our previous algorithms thus apply to
this general case, by executing them on each constraint set of the uninfluence
decomposition that is relevant to the task (namely, containing the variable $x$
to interpolate, or top-$k$ candidates from the selected variables $\calX_{\sigma}$):

\begin{corollaryrep}\label{cor:split}
    Given any constraint set $\calC$ and its uninfluence
    decomposition $\calC_1, \ldots, \calC_m$, assuming that each $\calC_i$ is a
    (reverse-)tree-shaped constraint set, we can solve the interpolation
    problem in time $O(\max_i \card{\mathsf{X}_i}^3)$ and the top-$k$
    problem in PTIME.
\end{corollaryrep}
\begin{proof}
        First, notice that any
	reverse-tree-shaped constraint set~$\calC$ can be transformed to a tree-shaped
	constraint set~$\calC'$ such that
	$\mathbb{E}_{\calC'}(x)= 1 - \mathbb{E}_\calC(x)$ for every $x \in \calX$,
	by reversing order constraints and replacing exact-value constraints
	$x=\alpha$ with $x= 1 - \alpha$. Second, we also observe that, formally,
        the constraint sets in the uninfluence decomposition may include
        variables with exact value constraints that are not connected to the
        tree-shaped structure, but it is obvious that these variables have no
        impact on the interpolation problem.

        Now, we use Lemma~\ref{lem:splittingb} to deduce that we can indeed
        solve the problems by solving them separately in each constraint set of the
        uninfluence partition. For the interpolation problem, we can compute the
        interpolated value of a variable by looking only at its uninfluence class. The complexity of top-$k$ follows.
\end{proof}

{On large tree-shaped taxonomies (e.g.,
the Google Product Taxonomy~\cite{googleproduct}), in an interactive setting
where we may ask user queries (e.g., the one in the Introduction), 
we can improve running times by asking more queries. Indeed, 
each answer 
about a category adds an exact-value constraint, and
reduces the size of the constraint
sets of the uninfluence decomposition, which decreases the overall running time,
thanks to the superadditivity of $x
\mapsto x^3$.
We do not study which
variables should be queried in order to reduce the running time of the algorithm;
see, e.g.,~\cite{parameswaran2011human} for tree-partitioning algorithms.}

\section{Other Variants}\label{sec:variants}
\begin{toappendix}
  \subsection{Proofs of Comparison Results}
  \label{apx:prfcompare}
\end{toappendix}
We have defined top-$k$ computation on constraint sets 
by considering
the \emph{expected value} of each variable under the uniform
distribution. {Comparing to different definitions of top-$k$ on
unknown values that have been studied in previous work, our definition has some
important properties~\cite{cormode2009semantics}}: it provides a ready
estimation for unknown values (namely, their expected value) and guarantees an
output of size~$k$. Moreover, it satisfies the \emph{containment property}
of~\cite{cormode2009semantics}, defined in our setting as follows:

\begin{definition}
	A top-$k$ definition satisfies the \deft{containment property} if for
	any constraint set\/ $\calC$ on variables $\calX$, for any predicate
	$\sigma$ (where we write $\calX_{\sigma}$ the selected variables), and for any $k <
	\card{\calX_{\sigma}}$, letting $S_k$ and $S_{k+1}$ be the ordered lists of top-$k$ and
	top-$(k+1)$ variables, 
	$S_k$ is a strict prefix of~$S_{k+1}$.
\end{definition}

The containment property is a natural desideratum: computing the top-$k$ for
some $k \in \mathbb{N}$
should not give different variables or order for the top-$k'$ with
$k' < k$. Our definition clearly satisfies the containment
property (except in the case of ties). By contrast, we will now review prominent
definitions of top-$k$ on uncertain data from related
work~\cite{soliman2007top,cormode2009semantics,zhang2009semantics},
and show that they
do not satisfy the containment property when we apply them to the possible world
distributions
studied in our setting. We focus on two prominent definitions,
\emph{U-top-$k$} and \emph{global-top-$k$}
and call our own definition \deft{local-top-$k$} when comparing to them; we
also discuss other variants in Appendix~\ref{apx:othervars}.

\subparagraph*{U-top-$\bm{k}$.}

The \emph{U-top-$k$} variant does not study \deft{individual} variables but
defines the output as the \emph{sequence} of $k$ variables most likely to be the top-$k$ (in that order), for the uniform distribution on $\pw(\calC)$. 
We call this alternative definition \deft{U-top-$k$} by
analogy with~\cite{soliman2007top,cormode2009semantics}. 
Interestingly, the U-top-$k$ and local-top-$k$ definitions sometimes disagree in
our setting:

\begin{lemmarep}\label{lem:localvsu}
  There is a constraint set~$\calC$ and selection predicate~$\sigma$
  such that local-top-$k$ and U-top-$k$ do not match, even
  for $k = 1$ and without returning expected values or probabilities.
\end{lemmarep}

\begin{proof}
 Let $\mu = 2/3$, $m = 1/\sqrt{2}$, and pick any $v$ such that $\mu < v < m$.
 Consider variables $x$, $x'$ and $y$, with the constraint set that imposes $x'
 \leq x$ and $y = v$. Fix $k = 1$ and consider the predicate $\sigma$ that
 selects all variables.
 It is immediate by linear interpolation that the expected
 value of $x$ is $\mu$. {Hence $y$, that has a greater expected value, is the local-top-1. 
 However, the marginal distribution of $x$ can easily be computed to be $p_x : t
  \mapsto 2t$. Intuitively, when $x$ becomes larger, it makes a larger range of
  values possible for $x'$, and thus also a larger range of possible worlds. 
 By integrating, we obtain that $x$ has a probability of~$0.5$ of exceeding $m$ -- a value greater than $v$, $y$'s value -- and hence its probability of being the top-1 is greater than $y$'s. Namely, $x$ is the U-top-1. Intuitively, the volume of possible worlds where $x$ is the top-1 is greater due to the asymmetry of $x$'s distribution. However, there is a larger \emph{gap} (on average) between $y$ and $x$ in worlds where $y$ is the top-1, which ultimately leads to a higher expected value for $y$.}
\end{proof}

We can easily design an algorithm to compute U-top-$k$ in
PSPACE and in polynomial time in the number of linear extensions of $\calC$:
compute the probability of each linear extension as in
Algorithm~\ref{alg:brute}, and then sum on linear extensions depending on which
top-$k$ sequence they realize (on the variables selected by $\sigma$), to
obtain the probability of each answer. Hence:

\begin{proposition}\label{prp:utopk}
  For any constraint set~$\calC$ over $\calX$, integer~$k$ and selection predicate~$\sigma$,
  the U-top-$k$ query for $\calC$ and $\sigma$ can be computed in PSPACE and in time
  $\oof{\mathrm{poly}(N)}$, where~$N$ is the number of linear extensions
  of~$\calC$.
\end{proposition}

Unlike Theorem~\ref{thm:expectmember}, however, this does not imply
$\fpsp$-membership: when
selecting the most probable sequence, 
the number of candidate sequences
may not be polynomial (as $k$ is not fixed). We leave to future work an
investigation of the precise complexity of U-top-$k$.

We show that in our setting
U-top-$k$ does not satisfy the \deft{containment
property} of~\cite{cormode2009semantics}.

\begin{lemmarep}\label{lem:utopkcp}
  There is a constraint set\/~$\calC$ without ties such that U-top-$k$ does not
  satisfy the containment property for the uniform distribution on
  $\pw(\calC)$.
\end{lemmarep}

\begin{proof}
  Consider variables $x_{\l}$, $x_{\h}$, $x^+_{\f}$, and $x^-_{\f}$, and the
  constraint set $\calC$ that imposes $x_{\l} \leq x_{\h}$, $x_{\f}^+ = .7$, and
  $x_{\f}^- = .69$.
  Consider the selection predicate $\sigma$ that selects all variables. The
  total volume of the constraint set is clearly $V = 1/2$.

  We first set $k = 1$. The first possible answer is $(x_{\f}^+)$ with probability
  $\frac{.7\times.7}{2 \cdot V} = .49$, and the second is $(x_{\h})$ with
  probability $.51$, so the U-top-$1$ is $(x_{\h})$.

  We then set $k = 2$. 
  There are four possible answers. The first possible
  answer is $(x_{\f}^+, x_{\f}^-)$ with probability $\frac{.69 \times .69}{2
  \cdot V} = .4761$.
  The second possible
  answer is $(x_{\h}, x_{\f}^+)$ with probability $\frac{.3 \times .7}{V} =
  .42$. The third possible answer
  is $(x_{\h}, x_{\l})$ with probability $.09$. The fourth possible answer is
  $(x_{\f}^+, x_{\h})$ with probability $\frac{.01\times .69 + .01 \times .01
  \times .5}{V} = .0139$. Hence, the U-top-$2$ is $(x_{\f}^+, x_{\f}^-)$.

  Hence, the U-top-$1$ variable does not occur in the U-top-$2$.
\end{proof}

\subparagraph*{Global-top-$\bm{k}$.}
We now study the \deft{global-top-$k$} definition~\cite{zhang2009semantics}, and
show that it does not respect the containment property either, even though it is
defined on individual variables:

\begin{definition}
  The \deft{global-top-$k$ query}, for a constraint set\/~$\calC$, selection
  predicate $\sigma$, and integer~$k$, returns the $k$ variables that have
  the highest probability in the uniform distribution on~$\pw(\calC)$ to be
  among the variables with the $k$ highest values, sorted by decreasing
  probability.
\end{definition}

\begin{lemmarep}\label{lem:gtopkcp}
  There is a constraint set\/~$\calC$ without ties such that global-top-$k$ does not
  satisfy the containment property for the uniform distribution on
  $\pw(\calC)$.
\end{lemmarep}

\begin{proof}
  Consider variables $x_\s$, $x_\f$, $x_\l$ and $x_\h$ and the constraint set~$\calC$
  that imposes $x_\l \leq x_\h$, $x_\l = .45$ and $x_\f = .73$. Consider
  $\sigma$ that selects all variables.

  Set $k \defeq 1$. Variable $x_\h$ has the highest value with probability
  $\frac{1}{V} (.73 \times
  (1-.73) + (1 - .73)^2 / 2)$. Variable $x_\f$ has the highest value with probability
  $\frac{1}{V} (.73-.45) \times .73$, which is less. The probability for $x_\s$
  is also less. So the global-top-1 is $(x_\h)$.

  Now, set $k \defeq 2$. Variable $x_{\h}$ has one of the two highest values in
  all cases except for $x_\f \geq x_\s \geq x_\h \geq x_\l$ and $x_\s \geq x_\f
  \geq x_\h \geq x_\l$, so it has one of the two highest values with
  probability $1 - \frac{1}{V} ((.73-.45)\times (1-.73) + (.73-.45)^2/2)$.
  However, variable $x_\f$ has one of the two highest values in all cases except
  for $x_\h \geq x_\s \geq x_\f \geq x_\l$ and $x_\s \geq x_\h \geq x_\f \geq
  x_\l$, so it has one of the two highest values with probability $1 - \frac{1}{V}
  (1-.73)^2$, which is more. Hence the first variable of the global-top-2
  is $x_\f$ and not $x_\h$.
\end{proof}

\begin{toappendix}
\subsection{Other Variants}
  \label{apx:othervars}
Additional variants of top-$k$ have been studied,
see~\cite{cormode2009semantics,zhang2009semantics}. However, in the context
of~\cite{cormode2009semantics}, these definitions do not satisfy the containment property either,
except for two. The first, \mbox{U-kRanks}~\cite{soliman2007top}, does not
satisfy the natural property that top-$k$ answers always contain $k$ different
variables. The second, expected ranks~\cite{cormode2009semantics}, resembles local-top-$k$ but uses ranks instead of values, so the definition is \deft{value-independent}. 
While this makes sense for top-$k$
queries designed to return tuples, as
in~\cite{cormode2009semantics}, we argue it is less sensible when
focusing on the numerical value of variables; this justifies our focus on
local-top-$k$.

Another possibility to define top-$k$ in our context, however, is to design it
based on different assumptions. One natural choice is to require a
\deft{stability} property, namely, adding exact-value constraints to fix some
variables to their interpolated values does not change the interpolated values
of the other variables. We can show that this property is not respected by our
scheme, but that it can be enforced on tree-shaped constraint sets: see Appendix~\ref{apx:stable}.
\end{toappendix}

\section{Related Work}
\label{sec:related}

We extend the discussion about related work from the Introduction.

\subparagraph*{Ranking queries over uncertain databases.}
A vast body of work has focused on providing semantics and evaluation
methods for order queries over uncertain databases, including 
top-$k$ and ranking queries (e.g.,~\cite{cormode2009semantics,detwiler2009integrating,HaghaniMA09,hua2011ranking,JestesCLY11,li2009ranking,re2007efficient,SolimanIB10,wang2011pruning,yi2008efficient}). 
Such works consider two main uncertainty types:
\emph{tuple-level uncertainty}, where the existence of tuples (i.e., variables) is
uncertain, and hence affects the query
results~\cite{cormode2009semantics,detwiler2009integrating,hua2011ranking,JestesCLY11,li2009ranking,re2007efficient,wang2011pruning,yi2008efficient};
and
\emph{attribute-level uncertainty}, more relevant to our problem, where the data tuples are known but
some of their values are unknown or
uncertain~\cite{cormode2009semantics,HaghaniMA09,JestesCLY11,SolimanIB10}.
Top-$k$ queries over uncertain data following \cite{SolimanIB10} was recently
applied to crowdsourcing applications in \cite{ciceri2016crowdsourcing}.
These studies are relevant to our work as they identify multiple possible
semantics for order queries in presence of uncertainty, and specify desired
properties for such
semantics~\cite{cormode2009semantics,JestesCLY11}; 
our
definition of top-$k$ satisfies the desiderata that are relevant to attribute-level uncertainty~\cite{JestesCLY11}. 

We depart from this existing work in two main respects. 
First, existing work assumes that each variable is given with an
\emph{independent function} that describes its probability distribution.
We do not assume this, and instead \emph{derive} expressions for the expected values of variables in a principled way from a uniform prior on the possible worlds.
Our work is thus well-suited to the many situations where probability
distributions on variables are not known, or where they are not independent
(e.g., when order constraints are imposed on them).
For this reason, the problems that we consider are generally computationally
harder. For instance, \cite{SolimanIB10} is perhaps the closest to our work,
since they consider the total orders compatible with given partial order
constraints. However, they assume independent marginal distributions, so they
can evaluate top-$k$ queries by only considering $k$-sized prefixes of the
linear extensions; in our setting even computing the top-1 element is hard
(Theorem~\ref{thm:topkhard}).

The second key difference is that other works \emph{do not try to estimate the top-$k$ values}, because they assume that the marginal distribution is given: they only focus on ranks.
In our context, we need to compute missing values, and need to account, e.g.,
for exact-value constraints and their effect on the probability of possible
worlds and on expected values (Section~\ref{sec:brute}).

We also mention our previous
work~\cite{amarilli2014uncertainty} 
which considers
the estimation of uncertain values (expectation and variance), 
but only in a \emph{total order}, and did not consider
complexity issues.

\subparagraph*{Partial order search.}
Another relevant research topic, \emph{partial order search}, considers queries
over elements in a partially ordered set to find
a subset of elements with a certain
property~\cite{amarilli2014complexity,davidson2013topk,faigle1986searching,gunopulos2003discovering,parameswaran2011human}.
This 
relates to many applications, e.g., crowd-assisted graph search~\cite{parameswaran2011human}, frequent itemset mining with the crowd~\cite{amarilli2014complexity}, and knowledge discovery, where the unknown data is queried via oracle calls~\cite{gunopulos2003discovering}.
These studies are
\emph{complementary to ours}: 
when the target
function can be phrased 
as a top-$k$ or interpolation problem,
if the search is stopped 
before all values are known, we can use our
method to estimate the complete output.

\subparagraph*{Computational geometry.}
Our work reformulates the interpolation problem as a centroid
computation problem in the polytope of possible worlds defined by the
constraint set. This problem has been studied independently by computational
geometry work~\cite{rademacher2007approximating,kannan1997random,maire2003algorithm}.

Computational geometry mostly studies \emph{arbitrary} convex polytopes (corresponding to polytopes defined by arbitrary linear constraint sets), and often
considers the task of \emph{volume computation}, which is related to the problem
of computing the centroid~\cite{rademacher2007approximating}. In this context,
it is known that computing the exact volume of a polytope is not in
$\fpsp$ because the output is generally not of polynomial size~\cite{lawrence1991polytope}. Nevertheless, several (generally exponential)
methods for exact volume computation \cite{bueler2000exact} have been developed.
The problem of approximation has also been studied, both theoretically
and practically
\cite{kannan1997random,simonovits2003compute,delorea2013software,cousins2015practical,LovaszD12,GeM15}.
Our problem of \emph{centroid} computation is studied
in \cite{maire2003algorithm}, whose algorithm is based on the idea of computing
the volume of a polytope by computing the lower-dimensional volume of its
facets. This is different from our algorithm, which divides the polytope
along linear extensions into subpolytopes, for which we apply a specific volume
and centroid computation method. 

Some works in computational geometry specifically study \emph{order polytopes}, i.e., the
polytopes defined by constraint sets with only order constraints and no
exact-value constraints. For such polytopes, volume computation is known to be
$\fpsp$-complete \cite{brightwell1991counting}, leading to a $\fpsp$-hardness result for
centroid computation \cite{rademacher2007approximating}. 
However, these results do not apply to \emph{exact-value constraints}, i.e.,
when order polytopes can only express order relations, between variables which are in $[0, 1]$.
Exact-value constraints are both highly relevant in practice (to represent
numerical bounds, or known information, e.g., for crowdsourcing), allow for
more general polytopes, and complicate the design of Algorithm~1, which must
perform volume computation and interpolation in each \emph{fragmented} linear
order.

Furthermore, to our knowledge, computational
geometry works do not study the top-$k$ problem, or polytopes that correspond to tree-shaped constraint sets, since these have no clear geometric interpretation.

\subparagraph*{Tree-shaped partial orders.}
Our analysis of tractable schemes for tree-shaped partial orders is reminiscent of the
well-known tractability of probabilistic inference in tree-shaped graphical
models~\cite{bishop}, and of the tractability of probabilistic query evaluation
on trees~\cite{cohen2009running} and treelike
instances~\cite{amarilli2015provenance}.
However, we study continuous distributions on numerical values, and the influence
between variables when we interpolate does not simply follow the tree structure; so 
our results do not seem to follow from these settings.

\section{Conclusion}
\label{sec:conc}

In this paper, we have studied the problems of top-$k$ computation and interpolation for data with unknown values and order constraints.
We have provided foundational solutions, including a general
computation scheme, complexity bounds, and analysis of tractable cases.

One natural direction for future work is to study whether our tractable cases
(tree-shaped orders, sampling) can be covered by more efficient PTIME
algorithms, or whether more general tractable cases can be identified:
for instance, a natural direction to study would be partial orders with a
\emph{bounded-treewidth} Hasse diagram, following recent tractability results
for the related problem of linear extension counting~\cite{kangas2016counting}. 
Another question is to extend our scheme to
request additional values from the crowd, as in~\cite{amarilli2014complexity,ciceri2016crowdsourcing},
and reduce the expected error on the interpolated values or top-$k$ query, relative to a user goal.
In such a setting, how should we choose which values to retrieve, and could we update incrementally
the results of interpolation when we receive new exact-value constraints?
Finally, it would be interesting to study whether our results generalize to different prior distributions on the polytope.

\subsection*{Acknowledgements}
This work is partially supported by the European Research Council under the FP7, ERC grant MoDaS, agreement 291071, by a grant from the Blavatnik Interdisciplinary Cyber Research Center, 
by the Israel Science Foundation (grant No.~1157/16), 
and by the Télécom ParisTech Research Chair on Big Data and Market Insights.

\bibliographystyle{abbrv}
\bibliography{bib}

\begin{toappendix}
\section{Alternative Interpolation Scheme}
\label{apx:stable}
We now consider variants of the interpolation problem, thus far performed by considering the expected value under the uniform distribution.
Since we are not aware of candidate variants in previous work, for interpolating over partial orders with exact-value constraints, we propose a new alternative variant. 
Rather than
imposing the connection to the uniform prior, we present natural desiderata for
an interpolation scheme on partial orders. We show that they are not respected by our
current definition, and show that a definition that respects them can be
proposed for tree-shaped partial orders (Definition~\ref{def:treeshaped}), it is in fact
unique, and it can be computed tractably.

We first define the notion of interpolation scheme.

\begin{definition}
  An \deft{interpolation scheme} is a function that maps any constraint set
  $\calC$ on variables $\calX$ to a mapping from $\calX$ to its interpolated value in $[0, 1]$.
\end{definition}

For instance, the interpolation scheme that we have studied thus far maps each variable to
  its expected value under the uniform distribution on $\pw(\calC)$. We refer to this scheme as \uniform in the sequel.

We define the first natural desideratum for interpolation schemes, stability: intuitively, an interpolation scheme is stable if assigning variables to
their interpolated value does not change the result of interpolation elsewhere. Formally,

\begin{definition}
An interpolation scheme $\calS$ is \deft{stable} if, for every constraint set $\calC$ over $\calX$ and every $x\in\calX$, $\calS$~assigns the same mapping $f:\calX\rightarrow[0,1]$ to both $\calC$ and $\calC\cup\{x=f(x)\}$.
\end{definition}

This property can be shown to be respected, e.g., by linear interpolation on total
orders. However, a counterexample shows that the stability property is not
respected by the uniform scheme:

\begin{lemma}\label{lem:unifnotstable}
  The \uniform scheme is not stable, even on tree-shaped constraint sets.
\end{lemma}

\begin{proof}
  Consider the set of variables $\{x_\r, x_\a, x_\b, x_\c, x_\d, x_\e\}$
  and the constraint set $\calC$ formed of the order constraints $x_\r \leq
  x_\a$, $x_\a \leq x_\b \leq x_\c$, $x_\a \leq x_\d \leq x_\e$, and the exact-value
  constraints $x_\r = 0$, $x_\c = .5$ and $x_\e = 1$. We can compute that the
  interpolated values for $x_\a$ and for $x_\b$ are $3/20$ and $13/40$
  respectively. However, adding the exact-value constraint $x_\b = 13/40$, the
  interpolated value for $x_\a$ becomes $611/4020$, which is different from
  $3/20$.
\end{proof}

Let us use stability as a guide to design a different interpolation scheme. We impose another desideratum to act as
a \emph{base case}, specifying what one expects from an interpolation scheme when
there is only a single unknown variable:

\begin{definition}
Let $\calC$ be a (non-contradictory) constraint set such that $x$ is the only unknown variable; $y_1,\dots,y_n$ are variables with exact-value constraints such that $y_i\leq x$; and $z_1,\dots,z_m$ are variables with exact-value constraints such that $x\leq z_i$ (having $n,m\geq 1$).
  We say an interpolation scheme is \deft{balanced} if, for each
  such~$\calC$, its interpolated value for $x$ is $\max_i(v(y_i)) + \min_i (v(z_i))
  \over 2$.
\end{definition}

In particular, the \uniform scheme is balanced in this sense.
However, we would like to find a scheme that is \emph{both balanced and stable}. For the case of general constraint sets, this problem remains open, and we leave it for future work. For tree-shaped constraint sets, we next not only show such a scheme, but also prove it is unique, as follows.

\begin{proposition}\label{prp:onestable}
  There is at most one interpolation scheme on trees that is both stable and
  balanced.
\end{proposition}

\begin{proof}
We first observe that, because of the balanced requirement, in any constraint
set $\calC$, for any unknown variable $x$ whose parent $y$ was interpolated to value $v$ and
whose children $z_i$ were interpolated to $w_i$, $x$ must have be interpolated to
$(v + (\min_i w_i))/2$. Indeed, considering the constraint set $\calC'$ where $y$
and $z_i$ have been set to those values, by stability, the interpolation value
of $x$ does not change. Hence, as the interpolation scheme is balanced, we
conclude that the claimed property holds.

We now show that the resulting set of equations always has at most one solution
on any tree-shaped constraint sets. Indeed, assume that there are two stable and
balanced interpolation
schemes $f$ and $g$ which yield different results on a tree-shaped constraint
set $\calC$. For all variables
$x$ of $\calC$, let $d(x) \defeq g(x) - f(x)$. Calling $x_{\r}$ the root
variable of~$\calC$, we must have $d(x_{\r}) \defeq 0$ for the root, because it
has an exact-value constraint by definition of tree-shaped constraint sets.

Now, as
$f$ and $g$ differ on $\calC$, there must be a variable $x$ where $d(x) \neq 0$.
Without loss of generality, we have
$d(x) > 0$. Hence, let us consider a variable $x$ with parent $y$ so that we have
$d(x) > d(y)$: as $d(x_{\r}) = 0$, we can find such a variable $x$ by picking a
variable which is as high as possible in the tree, such that $d(x) > 0$ but
$d(y) = 0$. Necessarily $x$ is not a leaf (as they have exact-value
constraints, so $d(x) = 0$), so $x$ has children. We show that $x$ has a child
$x_g$ such that $d(x_g) > d(x)$.

Consider $x_f$ the child of $n$ such that $f(x_f)$ is minimal
among children of~$x$, and $x_g$ defined in the analogous manner for $g$. Now,
as $f$ and $g$ are balanced, by our preliminary observation we have $f(x) = (f(y) + f(x_f))/2$ hence $f(x_f) =
2\cdot f(x) - f(y)$, and likewise $g(x_g) = 2\cdot g(x) - g(y)$.
But then, by minimality of $g(x_g)$,
we have $g(x_g) - f(x_f) \leq g(x_f) - f(x_f)$. Now, we have
$g(x_g) - f(x_f) = 2\cdot d(x) - d(y)$. Now, as we have $d(y) < d(x)$, we have
$d(x_f) > d(x)$, which is what we wanted to show.

Now, repeating the argument on $x_f$, we obtain a child $x_f^2$ of $x_f$ such
that $d(x_f^2) > d(x_f)$. Repeating the argument, we thus build a descending chain of
variables $x$ in the tree-shaped constraint set $\calC$ along
which $d$ is strictly increasing. When we reach the leaves, we obtain a
contradiction. This implies that we must have $d(x) = 0$ for all $x \in \calX$,
so that $f = g$ on $\calC$. Hence, there cannot be two different stable and
balanced interpolation schemes on tree-shaped constraint sets.
\end{proof}

We now prove the \emph{existence} of a stable and balanced interpolation scheme
on trees, which we dub \stable, and show that expected values under this scheme
can be computed in linear time:

\begin{proposition}\label{prp:stablebalance}
  There exists a stable and balanced interpolation scheme on tree-shaped
  constraint sets, and we can compute the interpolated values of \emph{all} variables
  according to this scheme in time $O(\card{\calX}^2)$.
\end{proposition}

\begin{proof}
  We compute the interpolation scheme on a tree-shaped constraint set $\calC$
  top-down. For each variable $x$ which has no exact-value constraint or
  interpolated value, but whose parent $y$ has an exact value or an interpolated
  value $v_y$, we consider all leaves $z$ which are descendants of $x$ (and have
  an exact-value constraint to some value $v_z$), and set the
  interpolated value of~$x$ to be the minimum of linear interpolation from $y$
  to $z$; namely, letting $d_x(z)$ be the depth of leaf $z$ in the subtree
  rooted at $x$, we set $v_x \defeq \min_z \left(v_y + \frac{v_z - v_y}{d_x(z)+1}\right)$.
  This can clearly be done in the indicated time bound.

  We now show that the resulting interpolation scheme is indeed balanced and
  stable. It is immediate to observe that this scheme is balanced. We now show that it
  is stable. Towards this, let us first show that for every variable $x$ with
  parent $y$, if $z$ is a leaf that achieved the minimum when interpolating $x$
  to its value, then for all variables on the path from $x$ to $z$, $z$ was also
  a leaf that achieved the minimum when interpolating their value. Indeed, it
  suffices to show the claim for the first variable $x'$ of this path, a child of
  $x$, and then repeat the argument.
  We know that, from our choice when interpolating $x$, by minimality of the
  interpolation result for $x$ using $z$,
  we have
  $\frac{v_z - v_y}{d_x(z)+1} \leq \frac{v_{z'} - v_y}{d_x(z')+1}$,
  where $v_y$ is the value of the parent $y$ of $x$; let us call the first
  quantity $\delta$ and the second $\delta'$. By definition of the interpolation
  of $x$, we then have $v_x \defeq v_y + \delta$.
  Now, consider any leaf $z'$ reachable from $x'$. We must show that $z$
  achieves the minimum when interpolating $x'$; in other words, we must compare
  $\eta \defeq \frac{v_{z} - v_x}{d_{x'}(z) + 1}$ and
  $\eta' \defeq \frac{v_{z'} - v_x}{d_{x'}(z') + 1}$,
  and show that $\eta \leq \eta'$; note that $d_{x'}(z) + 1 =
  d_x(z)$ and $d_{x'}(z') + 1 = d_x(y)$.
  The quantity $\eta$ can then be rewritten as $\frac{d_x(z)+1}{d_x(z)} \times
  \frac{1}{d_x(z)+1}(v_z - v_y - \delta)$, i.e., $\frac{d_x(z)+1}{d_x(z)} \times
  \left(\delta - \frac{\delta}{d_x(z)+1}\right)$, which simplifies to $\delta$:
  hence, $\eta = \delta$.
  The quantity $\eta'$ can be written as $\frac{d_x(z') + 1}{d_x(z')} \times
  \frac{1}{d_x(z')+1} \left(v_{z'} - v_y - \delta\right)$, i.e.,
  $\frac{d_x(z') + 1}{d_x(z')} \times \left(\delta' -
  \frac{\delta}{d_x(z')+1}\right)$, which simplifies to $\frac{(d_x(z')+1)
  \delta' - \delta}{d_x(z')}$. Now, as $\delta' \geq \delta$, we deduce that
  $\eta' \geq \frac{(d_x(z')+1) \delta - \delta}{d_x(z')}$, so that
  $\eta' \geq \delta$. Hence, we have $\eta' \geq \eta$, so that the leaf $z$ also
  achieves the minimum for variable $x'$. Repeating the argument on the path from $x$ to
  $z$, we have shown the claim.

  From this initial claim, we deduce the following (*): for any variable $x$, letting $z$ be a
  leaf that achieves the minimum when interpolating $x$ (once the value of its
  parent $y$ is known), then the variables in the path from $y$ to~$z$ are
  interpolated according to linear interpolation on that chain. This is
  immediate by the previous claim, as all variables on that path are interpolated
  using linear interpolation from their parent to that same leaf (or another
  minimal leaf that sets them to the same value).

  We now show a similar auxiliary claim. Let us define, once we have interpolated
  in $\calC$, the function $u$ that maps each non-root variable $x$ to $u(x)$
  defined as the interpolated value of $x$ minus that of its parent. We show
  that (**): for any variables $y$, $x$, $x'$, where $y$ is the parent of $x$ and $x$
  is the parent of $x'$, then $u(x) \leq u(x')$.
  Indeed, let $v_y$, $v_x$ and $v_{x'}$ be the interpolated values, and
  let $z'$ be the witness
  leaf used to interpolate for $x'$. By
  definition of the scheme, we have $u(x') = \frac{v_{z'} - v_x}{d_{x'}(z') +
  1}$. Furthermore,
  letting $z$ be the witness leaf used to interpolate for $x$,
  we have $u(x) = \frac{v_z - v_y}{d_x(z)+1}$. Using the notation above,
  note that $u(x')=\eta'$ and $u(x)=\delta$. By the same reasoning as for
  claim (*) to show $\delta=\eta \leq \eta'$, we conclude that $u(x) \leq u(x')$.

  We are now ready to show that the scheme is stable.
  Consider the initial tree-shaped constraint set $\calC$, and
  let us set a variable $x$ to its interpolated value $v_x$, yielding $\calC'$.
  Note that $\calC'$ is no longer tree-shaped, but it can be rewritten by
  Lemma~\ref{lem:splittingb} to two tree-shaped constraint sets.
  It is then clear that all variables that are descendants of $x$ in $\calC$ are
  interpolated in the same manner in $\calC'$ and in $\calC$, as the scheme
  proceeds top-down and the value of $x$ in $\calC'$ is by definition the same
  as its interpolated value in $\calC$. We now show that the ancestors of $x$ in
  $\calC$ are interpolated in the same way in $\calC'$ than in $\calC$, which is
  clearly sufficient to justify the claim that \emph{all} variables in $\calC'$ are
  interpolated in the same way as in $\calC$. Let us therefore pick an ancestor $x'$ of
  $x$, which is neither $x$ nor the root variable, otherwise the claim is
  trivial; we pick it as high as possible in the tree, so the interpolated value
  $v_y$ of its ancestor~$y$ is the same in $\calC$ and in~$\calC'$.

  We first show that the interpolated value for $x'$ in $\calC'$ is no higher
  than in $\calC$. Assuming to the contrary that it is, then it must be the case
  that $x'$ was interpolated in $\calC$ using as minimal leaf $z$ some leaf which
  is a descendant of $x$ in~$\calC$, as otherwise we can still interpolate using
  $z$ in $\calC'$ and obtain the same result. Now, if $x'$ was interpolated in
  $\calC$ using $z$ as minimal leaf, then, by our preliminary claim (*), $x$ was
  interpolated in $\calC$ following linear interpolation between the parent $y$
  of $x'$ and the leaf $z$. Hence, using the new leaf $x$ in $\calC'$ to
  interpolate $x'$ in $\calC'$ yields the same result as the interpolation in
  $\calC$. Contradiction.

  Second, we show that the interpolated value for $x'$ in $\calC'$ is no lower
  than in $\calC$. Assuming to the contrary that it is, then, if $x'$ was
  interpolated in $\calC'$ following a leaf $z$ which is not $x$, then we
  immediately reach a contradiction as we should have used the same leaf $z$ to
  interpolate to the same value in $\calC$. Hence, we must have interpolated
  $x'$ in $\calC'$ using the new leaf $x$, and $x'$ was interpolated in $\calC'$
  following linear interpolation between $y$ and $x$. Let $\gamma$ be the value
  difference between two consecutive nodes in $\calC'$ on this path, and $l$ the
  length of the path. Calling~$u(x)$
  for a variable~$x$ the difference between its interpolated value and the value
  of its parent in $\calC$, we must then have $u(x') > \gamma$, because the
  value of $x'$ in $\calC$ is strictly greater than in $\calC'$, and $y$ has
  same value in $\calC$ and $\calC'$. By preliminary claim (**), we have reached a contradiction,
  because then the function $u$ always takes values which are $> \gamma$ on the
  path from $x'$ to $x$ in~$\calC$, so that when we reach $x$ we
  know that the value of $x$ in $\calC$ is $> l \cdot \gamma$, contradicting the fact that
  it is $l \cdot \gamma$, as we know from $\calC'$.
\end{proof}

\end{toappendix}

\end{document}